\documentclass[a4paper,11pt,reqno]{amsart}
\usepackage[margin=1in,includehead,includefoot]{geometry}
\usepackage[utf8]{inputenc}
\usepackage[T1]{fontenc}
\usepackage{lmodern,microtype}
\usepackage[hidelinks]{hyperref}
\usepackage{booktabs}
\usepackage{wrapfig}
\usepackage{mathtools,amssymb,amsthm,amsmath}
\usepackage{stmaryrd}
\usepackage{graphicx}
\usepackage[table]{xcolor}
\usepackage{wrapfig}
\usepackage{enumitem}
\usepackage[textsize=scriptsize,textwidth=2cm]{todonotes}
\usepackage{mdframed}
\usepackage{bm}
\usepackage{xpatch}
\usepackage{multirow}
\usepackage{hhline}
\usepackage{dsfont}
\usepackage{hyperref}
\usepackage{xcolor}
\usepackage{bookmark}
\usepackage{cases}
\usepackage{multirow}
\usepackage[margin=0mm,font=scriptsize,format=plain]{caption}
\usepackage[subrefformat=parens,labelformat=parens,justification=centering]{subcaption}
\usepackage{adjustbox}
\usepackage{ragged2e}
\usepackage{xspace}
\usepackage[citestyle=numeric,style=numeric,maxbibnames=99,natbib=true]{biblatex}
\usepackage[noend,ruled,resetcount]{algorithm2e} 
\usepackage{tikz-network}
\usepackage{blkarray}
\usepackage{pgfplots}
\usepackage{todonotes}
\usepackage{cleveref}
\usepackage{crossreftools}
\usepackage{makecell}

\pgfplotsset{compat=1.10}
\usepgfplotslibrary{fillbetween}
\usetikzlibrary{patterns}
\usetikzlibrary{decorations.pathreplacing,calligraphy}
\usetikzlibrary{decorations.pathmorphing}
\usetikzlibrary{backgrounds}
\usetikzlibrary{arrows}
\usetikzlibrary{calc}
\usetikzlibrary{plotmarks}
\usetikzlibrary{fadings}
\usetikzlibrary{matrix,decorations.pathreplacing,positioning,fit}
\usetikzlibrary{backgrounds}
\usetikzlibrary{arrows}
\usetikzlibrary{plotmarks}
\usetikzlibrary{fadings}

\vfuzz \hfuzz
\newtheorem{theorem}{Theorem}
\newtheorem{corollary}[theorem]{Corollary}
\newtheorem{lemma}[theorem]{Lemma}
\newtheorem{proposition}[theorem]{Proposition}
\theoremstyle{remark}

\crefname{claim}{Claim}{Claims}
\crefname{lemma}{Lemma}{Lemmas}
\crefname{theorem}{Theorem}{Theorems}
\crefname{part}{part}{parts}
\crefname{alt}{alternative}{alternatives}
\crefname{step}{step}{steps}
\crefname{algocf}{Algorithm}{Algorithms}
\crefname{eq}{equality}{equalities}
\crefname{eqs}{equalities}{equalities}
\crefname{eqn}{equation}{equations}
\crefname{ineq}{inequality}{inequalities}
\crefname{ineqs}{inequalities}{inequalities}
\crefname{exp}{expression}{expressions}
\crefname{prop}{property}{properties}
\crefname{props}{properties}{properties}
\crefname{imp}{implication}{implications}
\crefname{eqv}{equivalence}{equivalences}
\crefname{def}{definition}{definitions}

\creflabelformat{eq}{(#2#1#3)}
\creflabelformat{eqs}{(#2#1#3)}
\creflabelformat{eqn}{(#2#1#3)}
\creflabelformat{ineq}{(#2#1#3)}
\creflabelformat{ineqs}{(#2#1#3)}
\creflabelformat{exp}{(#2#1#3)}
\creflabelformat{prop}{(#2#1#3)}
\creflabelformat{props}{(#2#1#3)}
\creflabelformat{imp}{(#2#1#3)}
\creflabelformat{eqv}{(#2#1#3)}
\creflabelformat{def}{(#2#1#3)}

\newcommand{\calG}{\mathcal{G}}
\newcommand{\calI}{\mathcal{I}}

\newcommand{\calP}{\mathcal{P}}
\newcommand{\calQ}{\mathcal{Q}}

\newcommand{\cI}{\mathcal{I}}

\newcommand{\RR}{\mathbb{R}}
\newcommand{\NN}{\mathbb{N}}
\newcommand{\EE}{\mathbb{E}}
\newcommand{\PP}{\mathbb{P}}
\newcommand{\bfp}{p}
\newcommand{\bfP}{P}
\newcommand{\bfq}{q}
\newcommand{\bfs}{s}
\newcommand{\bfx}{x}
\newcommand{\bfz}{z}
\newcommand{\bflambda}{\lambda}
\newcommand{\ie}{i.e.,\xspace}

\newcommand{\indfam}{\calI}
\newcommand{\indins}{\Psi}
\newcommand{\MEC}{\normalfont\textsc{M}}
\newcommand{\OPT}{\normalfont\textsc{Opt}}

\DeclareMathOperator{\spn}{span}
\DeclareMathOperator{\score}{\sigma}
\newcommand{\nalgtwopar}{{\normalfont\textsc{Partition}}}
\newcommand{\algtwopar}{{\normalfont\textsc{Par}}}
\newcommand{\nalgprukna}{{\normalfont\textsc{KnapsackPlurality}}}
\newcommand{\algprukna}{{\normalfont\textsc{KPR}}}

\newcommand{\algkgreedy}{{\normalfont\textsc{KGr}}}

\newcommand{\algkgreedyex}{{\normalfont\textsc{EKGr}}}
\newcommand{\nalgpruknadet}{{\normalfont\textsc{DetKnapsackPlurality}}}
\newcommand{\algpruknadet}{{\normalfont\textsc{DKPR}}}
\newcommand{\nalgprugen}{{\normalfont\textsc{MatroidPlurality}}}
\newcommand{\algprugen}{{\normalfont\textsc{MPR}}}

\newcommand{\nalgvertex}{{\normalfont\textsc{VertexPartition}}}
\newcommand{\algvertex}{{\normalfont\textsc{VP}}}

\newcommand{\nalgmatroidpar}{{\normalfont\textsc{BestFromPartition}}}
\newcommand{\algmatroidpar}{{\normalfont\textsc{BP}}}

\definecolor{color0}{RGB}{255,217,50}
\definecolor{color1}{cmyk}{1, .37, 0, 0}
\colorlet{color1trans}{color1!60!white}
\definecolor{color2}{cmyk}{0, .50, .95, .05}
\definecolor{color3}{cmyk}{.18, .98, .18, 0}
\colorlet{color3trans}{color3!80!white}

\bibliography{bibliography.bib}

\begin{document}
\title{Impartial Selection Under Combinatorial Constraints}
\author[J.~Cembrano]{Javier Cembrano}
\address[J.~Cembrano]{Department of Algorithms and Complexity, Max Planck Institut für Informatik, Germany}
\email{\href{mailto:jcembran@mpi-inf.mpg.de}{jcembran@mpi-inf.mpg.de}}
\urladdr{\url{https://sites.google.com/view/javier-cembrano/}}

\author[M.~Klimm]{Max Klimm}
\address[M.~Klimm]{Institute for Mathematics, Technische Universität Berlin, Germany}
\email{\href{mailto:klimm@math.tu-berlin.de}{klimm@math.tu-berlin.de}}
\urladdr{\url{https://www3.math.tu-berlin.de/disco/team/klimm/}}

\author[A.~Merino]{Arturo Merino}
\address[A.~Merino]{Institute of Engineering Sciences, Universidad de O'Higgins, Chile}
\email{\href{mailto:arturo.merino@uoh.cl}{arturo.merino@uoh.cl}}
\urladdr{\url{https://amerino.cl}}

\begin{abstract}
Impartial selection problems are concerned with the selection of one or more agents from a set based on mutual nominations from within the set.
To avoid strategic nominations of the agents, the axiom of impartiality requires that the selection of each agent is independent of the nominations cast by that agent.
This paper initiates the study of impartial selection problems where the nominations are weighted and the set of agents that can be selected is restricted by a combinatorial constraint.
We call a selection mechanism \mbox{$\alpha$-optimal} if, for every instance, the ratio between the total sum of weighted nominations of the selected set and that of the best feasible set of agents is at least~$\alpha$.
We show that a natural extension of a mechanism studied for the selection of a single agent remains impartial and \mbox{$\frac{1}{4}$-optimal} for general independence systems, and we generalize upper bounds from the selection of multiple agents by parameterizing them by the girth of the independence system.
We then focus on independence systems defined by knapsack and matroid constraints, giving impartial mechanisms that exploit a greedy order of the agents and achieve approximation ratios of~$\frac{1}{3}$ and~$\frac{1}{2}$, respectively, when agents cast a single nomination. For graphic matroids, we further devise an impartial and \mbox{$\frac{1}{3}$-optimal} mechanism for an arbitrary number of unweighted nominations.
\end{abstract}

\maketitle 

\section{Introduction}
The organization of science itself involves several situations that require impartial judgment and review by peers. A prominent example is the selection process for publication venues, such as the proceedings of the \emph{Conference on Web and Internet Economics (WINE)}. For this and similar events, it is customary that there is a substantial overlap between the set of program committee members and authors of papers. Program committee members provide reviews for papers that (in a slight simplification of the matter) can be viewed as assigning a numerical score, and the goal is to select papers with an overall high score for the conference subject to a limit on the number of papers that can be accepted.
In a similar vein, \citet{merrifield2009telescope} proposed to allot telescope time to researchers based on a peer review of the proposals among the researchers that submitted proposals themselves, a mechanism that was later applied by the Sensors and Sensing Systems program of the National Science Foundation (NSF).

Such processes of peer review are clearly an important pillar of the scientific method, but in situations when people both provide and receive reviews, it suffers from the fact that providing particularly low reviews for others may increase the chances of the own papers or projects being successful.
For the allocation of telescope time, \citeauthor{merrifield2009telescope} address this concern by providing extra benefits for researchers whose reviews coincide with other reviews. This method, however, comes with further downsides since the original goal should be to elicit the projects most worthy of the telescope time, which may not necessarily coincide with the projects submitted by the most accurate reviewers. Moreover, while it may motivate some reviewers to submit honest reviews, it does not provide any theoretical guarantees in this regard.

To prevent incentive issues like these, \citet{alon2011sum} and \citet{holzman2013impartial} introduced the concept of the \emph{impartiality} of a selection mechanism. While they studied this concept only for the selection of a single entity or a fixed number of $k$ entities, respectively, and only based on binary scores, it is straightforward to generalize this concept to more general selection problems in the following way.
Formally, consider a finite set of agents $E$ that cast scores for each other, i.e., for every $i,j \in E$ with $i \neq j$ the value $w_{ij} \geq 0$ is the score assigned by agent~$i$ to agent~$j$. A randomized selection mechanism takes the agents' mutual scores as input and returns a lottery over subsets of agents.
A selection mechanism is called \emph{impartial} if the probability of selection of an agent is independent of the scores cast by this agent for other agents.
Put differently, impartiality requires that there are no two inputs to the mechanism that differ only in the scores cast by agent~$i$ and where the probability of selection of agent~$i$ for the two inputs differ.
Impartiality is clearly an appealing concept for peer review applications since it provides the theoretical guarantee that no agent may tamper with its own selection by manipulating the scores it submits.
Unfortunately, it is easy to convince ourselves that straightforward mechanisms that simply select a subset of agents with the highest total score are not impartial. This makes it natural to study the approximation guarantees of impartial mechanisms. For $\alpha \in [0,1]$, we call a selection mechanism \emph{$\alpha$-optimal} if, for every instance, the ratio between the (expected) total score of the selected agents and the total score of any feasible set of agents is at least $\alpha$.
We then call $\alpha$ the \emph{approximation ratio} of the mechanism.
Previous research established bounds on the best-possible values of $\alpha$ for impartial mechanisms for the selection of a single agent \cite{alon2011sum,cembrano2023single,fischer2015optimal,holzman2013impartial,} and up to $k$ agents \cite{bjelde2017impartial,cembrano2023deterministic} for a fixed number $k \in \mathbb{N}$. 
These results, except the one of \citet{cembrano2023deterministic}, only apply in the case where the scores submitted by the agents are binary, i.e., either $0$ or $1$.

In this paper, we initiate the study of selection mechanisms where the set of agents that can be selected is given by an \emph{independence system} $(E,\mathcal{I})$, i.e., $\mathcal{I} \subseteq 2^E$ and $T \in \mathcal{I}$ whenever $T \subseteq S$ with $S \in \mathcal{I}$.
Depending on the application, the independence system may encode various feasibility constraints. 

First, consider a stylized version of the conference proceedings problem where only single-author papers are submitted and every author submits at most one paper. Assume that at most $k$ papers can be selected. This can be formulated as an impartial selection problem where authors and program committee members correspond to agents and the independence system is a \emph{$k$-uniform matroid}, i.e., $\mathcal{I} = \{S \subseteq E \mid |S| \leq k\}$.
This setting is also called the \emph{$k$-selection problem}.
If the conference features several tracks $T$, papers are submitted to tracks, and each track $t \in T$ can accept at most $k_t \in \mathbb{N}$ papers, the problem can be modeled as a \emph{partition matroid}.
Another example involves an election designed to select a subset of agents and assign each a specific chore or position, with each agent being compatible with only a certain set of chores. 
In this context, it is logical to require that each chore is assigned to no more than one agent, i.e., that the chosen subset of agents can be matched to a subset of chores in the graph representing the compatibility relations through its edges. This feasibility relation is captured by a \emph{transversal matroid}.
Finally, for the allocation of telescope time, it is reasonable to assume that there is a total number of $C \in \mathbb{N}$ time slots available and the proposal submitted by researcher~$i$ requires $s_i \in \mathbb{N}$ time units of the total time available. Then, the set of researchers whose proposals can be accepted can be modeled with a \emph{knapsack constraint}, so that $\mathcal{I} = \{S \subseteq E \mid \sum_{i \in S } s_i \leq C\}$.

\subsection{Our Contribution and Techniques}
This paper studies impartial selection problems that generalize the existing models in two ways. First, we consider general combinatorial constraints given by an independence system. These combinatorial constraints contain the previously studied settings as a special case when the independence system is a uniform matroid. Second, we allow agents to submit weighted nominations (or \emph{scores}) for each other while previous research mainly concentrated on the case of binary nominations that are either cast (corresponding to a score of $1$) or not (corresponding to a score of $0$).
To obtain a more fine-grained set of results, we further introduce the \emph{sparsity} of a voting setting. 
For $d\in \mathbb{N}$, we call a voting setting \emph{$d$-sparse} if only instances where every agent casts scores for at most $d$ other agents are considered. 
Our main results are summarized in Table~\ref{tab:results}.
{\small
\begin{table}[tb]
\makebox[0cm]{
\begin{tabular}{@{}l @{\quad} c @{~} l @{\ \ } c @{~} l @{\ } c @{~} l @{\quad} r @{~} l @{\ \ } r @{~} l @{ } r @{~} l @{}}
\toprule \\[-8pt]
            & \multicolumn{6}{c}{\raisebox{0.3em}{\rule{1.6cm}{0.4pt}}\;\textbf{Independence systems}\;\raisebox{0.3em}{\rule{1.6cm}{0.4pt}}}                                                                       & \multicolumn{2}{c}{\textbf{Knapsack}\ \ } & \multicolumn{4}{c}{\raisebox{0.3em}{\rule{1.1cm}{0.4pt}}\;\textbf{Matroids}\;\raisebox{0.3em}{\rule{1.1cm}{0.4pt}}}                                      \\[2pt]
            & \multicolumn{2}{c}{general \ } & \multicolumn{2}{c}{w/girth $g$, 1-sparse} & \multicolumn{2}{c}{w/girth $g$\ \ } & \multicolumn{2}{c}{1-sparse\ \ } & \multicolumn{2}{c}{1-sparse\ \ } & \multicolumn{2}{c}{\makecell{simple graphic\\ (binary)}} \\[3pt] \midrule\\[-7pt]
\textbf{Lower} & \multirow{2}{*}{\normalsize$\frac{1}{4}$}       &  \multirow{2}{*}{(Thm.~\ref{thm:gral-indep-lb})}     & \multirow{2}{*}{\normalsize$\frac{1}{4}$}                   &    \multirow{2}{*}{(Thm.~\ref{thm:gral-indep-lb})}        & \multirow{2}{*}{\normalsize$\frac{1}{4}$}            &  \multirow{2}{*}{(Thm.~\ref{thm:gral-indep-lb})}     & \multirow{2}{*}{\ \normalsize$\frac{1}{3}$}       &   \multirow{2}{*}{(Thm.~\ref{thm:knapsack-lb})}     & \multirow{2}{*}{\ \normalsize$\frac{1}{2}$}       &   \multirow{2}{*}{(Thm.~\ref{thm:matroids-lb})}    & \multirow{2}{*}{\ \ \ \normalsize$\frac{1}{3}$}          &   \multirow{2}{*}{(Thm.~\ref{thm:graphic-matroids-lb})} \\
\textbf{bound} & & & & & & & & & & & & \\[5pt]
\textbf{Upper} & \multirow{2}{*}{\normalsize$\frac{1}{2}$}       & \multirow{2}{*}{(Prop.~\ref{prop:ub-gral})}      & \multirow{2}{*}{\ \ \normalsize$\frac{g(g-1)-1}{g(g-1)}$}            & \multirow{2}{*}{(Thm.~\ref{thm:gral-indep-ub-plu})}           & \multirow{2}{*}{\ \ \normalsize$\frac{2g}{2g+1}$}       &  \multirow{2}{*}{(Thm.~\ref{thm:gral-indep-ub-app})}     & \multirow{2}{*}{\normalsize$\frac{1}{2}$}       & \multirow{2}{*}{(Prop.~\ref{prop:ub-gral})}       & \multirow{2}{*}{\normalsize$\frac{1}{2}$}       &  \multirow{2}{*}{(Prop.~\ref{prop:ub-gral})}      & \multirow{2}{*}{\normalsize$\frac{5}{6}$}          & \multirow{2}{*}{(Thm.~\ref{thm:gral-indep-ub-plu})} \\
\textbf{bound} & & & & & & & & & & & & \\[3pt]
\bottomrule       
\end{tabular}
}
\caption{
Our main results.\label{tab:results}
}
\end{table}
}

We first consider the setting where selection constraints are defined by general independence systems. We show that the~$2$-partition mechanism proposed by \citet{alon2011sum} still ensures an approximation ratio of~$\frac{1}{4}$ in this broader context. Additionally, we provide upper bounds on the approximation ratio achievable by impartial mechanisms parameterized by the \emph{girth}~$g \coloneqq g(\mathcal{I})$ of the independence system $\mathcal{I}$, defined as the cardinality of the smallest dependent set, i.e., $g(\mathcal{I}) \coloneqq \min\{|S| \mid S \in 2^E \setminus \mathcal{I}\}$.
Specifically, we derive an upper bound of~$1-\frac{1}{g(g-1)}$ for $1$-sparse instances, and a general upper bound of~$1-\frac{1}{2g+1}$. 
Asymptotically, these bounds extend the respective best-known upper bounds of~$1-\frac{2}{(k+1)^2}$ and~$1-\frac{1}{k+2}$ for~$k$-selection (where~$g=k+1$).

We then move on to the natural case of knapsack constraints, where each agent~$i \in E$ has size~$s_i$ and feasible sets consist of subsets with a total size below a specified capacity~$C$. 
For $d$-sparse instances, we devise a randomized mechanism with an approximation ratio of~$\frac{1}{d+2}$, implying a \mbox{\smash{$\frac{1}{3}$}-approximation} in the case of each agent evaluating at most one other agent.
We follow the classic greedy approach and sort the agents in decreasing order of their ratio of score and size, and call the union of the greedy solution and the fractional agent (if such an agent exists) the \textit{extended greedy solution}.
We then take inspiration from a mechanism proposed by \citet{tamura2014impartial} to carefully construct a family of up to~$d+2$ feasible sets with two properties: (1)~No agent can affect the fact of belonging or not to some set in the family by changing its votes, and (2) all agents in the extended greedy solution belong to some set in the family.
Given these two properties, the mechanism that selects each set in the family with probability~$\smash{\frac{1}{d+2}}$ returns a feasible set, is impartial, and \mbox{\smash{$\frac{1}{d+2}$}-optimal}. While deterministic mechanisms cannot provide any constant approximation ratio better than~$0$, as instances with agents of size equal to~$C$ correspond to the~$1$-selection setting for which no non-trivial results are possible \cite{alon2011sum,holzman2013impartial}, positive results can be achieved for restricted instances with relatively small agents. Specifically, we introduce a deterministic impartial mechanism with a multiplicative guarantee of~$1-\frac{s_{\max}}{C}(d+1)$ on $d$-sparse instances where the largest size is~$s_{\max}$. Once again, it builds upon the mechanism by \citet{tamura2014impartial}, this time considering a reduced capacity of~$C-d\cdot s_{\max}$ to ensure feasibility in a deterministic manner.

We finally turn our attention to matroids.
We devise a randomized mechanism for impartial selection under matroid constraints that reaches an approximation ratio of~$\frac{1}{2}$ on $1$-sparse instances. Again based on the mechanism by \citet{tamura2014impartial}, we greedily select with probability~$\frac{1}{2}$ a maximal independent set of agents sorted by their score, along with all agents that could turn themselves into this set by modifying their votes. While impartiality and approximate optimality follow easily, the feasibility of this mechanism is proven by appropriately combining basis-exchange and optimality properties of matroids. Remarkably, this mechanism achieves an optimal general guarantee among impartial mechanisms, as an upper bound of~$\frac{1}{2}$ already holds for~$1$-selection. We then focus on \textit{simple graphic matroids}, one of the most classic classes of matroids where agents represent edges of a simple undirected graph and the independent sets correspond to forests. 
We give an impartial and \mbox{$\frac{1}{3}$-optimal} mechanism for this setting in the case of binary scores. 
In a first step, we partition the edges into sets by taking a permutation of the vertices uniformly at random and adding to one set all incident edges to the vertices that have not been added to another set. In a second step, we run the \textit{$k$-partition with permutation} mechanism for impartial $k$-selection by \citet{bjelde2017impartial}, which ensures the selection of one edge from each set with a maximum observed score, i.e., score from other sets and score from edges in the same set that appear before in a permutation of the edges taken uniformly at random.

\subsection{Related Work}
Impartial $1$-selection has been studied by \citet{holzman2013impartial}, who showed that deterministic mechanisms cannot satisfy a set of natural axioms and thus cannot be $\alpha$-optimal for any $\alpha > 0$. 
\citet{alon2011sum} showed that also selection mechanisms that always select exactly $k$ agents cannot be $\alpha$-optimal for $\alpha > 0$. 
They 
devised the randomized $2$-partition mechanism, showed that it is $\smash{\frac{1}{4}}$-optimal, and gave an example showing that no $1$-selection mechanism can be better than $\smash{\frac{1}{2}}$-optimal.
A mechanism with a matching approximation ratio of $\smash{\frac{1}{2}}$ was then given by \citet{fischer2015optimal}.
\citet{bousquet2014near} gave a mechanism whose approximation ratio approaches $1$ as the maximum number of votes for an agent goes to infinity.
\citet{bjelde2017impartial} studied the approximation ratios of deterministic and randomized $k$-selection mechanisms. Deterministic $k$-selection mechanisms with improved approximation guarantees were given by \citet{cembrano2023deterministic}.

\citet{caragiannis2022impartial} initiated the study of impartial mechanisms with additive approximation guarantees. 
Specifically, they proposed a randomized $1$-selection mechanisms for $1$-sparse binary scores with an additive approximation guarantee of $O(\sqrt{n})$, i.e., where the difference between the maximal number of votes cast for any agent and the expected number of votes cast for the selected agent is at most $O(\sqrt{n})$, where $n$ is the number of agents.
\citet{cembrano2024impartial} obtained the same additive guarantee for a deterministic mechanism. They further show that without the restriction to $1$-sparse instances only the trivial guarantee of $n-1$ is possible. \citet{caragiannis2023impartial} studied additive guarantees of mechanisms that have access to prior information on the expected votes that agents will receive.

Axiomatic studies of impartial mechanisms complementing the results of \citeauthor{holzman2013impartial} have been conducted by \citet{mackenzie2015symmetry} for symmetric mechanisms, by \citet{tamura2014impartial,tamura2016characterizing} for mechanisms selecting more than one agent, and by \citet{mackenzie2020axiomatic} for selection mechanisms that satisfy additional fairness properties. \citet{cembrano2022optimal} studied the approximation guarantees of $k$-selection mechanisms inspired by those proposed by \citet{tamura2014impartial} in more detail.
Impartial peer review has been studied by \citet{aziz2019strategyproof,kurokawa2015impartial,mattei2021peernomination}. While \citeauthor{aziz2019strategyproof} as well as \citeauthor{mattei2021peernomination} do not obtain theoretical guarantees comparable to ours, a mechanism proposed by \citeauthor{kurokawa2015impartial} achieves an approximation ratio of $\frac{k}{k+d}$ when up to $k$ papers are selected and the instance is $d$-sparse. 

For the basic definition and results on the knapsack problem, we refer the reader to the books by Kellerer, Pferschy, and Pisinger~\citep{kellerer2004knapsack} and \citet[Chapter~17]{korte18}.
For an overview of matroid theory and its various applications, we refer the reader to the books by \citet{oxley1992matroid,welsh2010matroid}.

\section{Preliminaries}
\label{sec:prelims}

For~$n\in \NN$, we let $[n]\coloneqq\{1,\ldots,n\}$ denote the set of integers from~$1$ to~$n$. 
We use the Iverson bracket~$\llbracket \cdot \rrbracket$ for the indicator function, i.e.,~$\llbracket p \rrbracket =1$ if the logical proposition~$p$ is true and~$\llbracket p \rrbracket =0$ otherwise. 
For a set~$E$, we assume that it is totally ordered; i.e., $E = \{e_1,\dots, e_m\}$.
This total order extends naturally to a total order~$\prec_E$ over~$2^E$ by mapping each set~$S \subseteq E$ to a bitstring~$x_S$ of length~$E$ where the~$i$-th coordinate is~$\llbracket e_i \in S \rrbracket$ and ordering lexicographically; i.e.,~$S \prec_E T$ if and only if~$x_S$ is smaller lexicographically than~$x_T$.%
\footnote{This total order plays the role of a tie-breaking rule throughout the paper.}
For a function~$\phi\colon E\to \RR$, we define the total order~$\prec_\phi$ over~$2^E$ by
\[
    S \prec_\phi T \quad \Longleftrightarrow \quad \phi(S)<\phi(T) \text{ or } (\phi(S)=\phi(T) \text{ and } S\prec_E T), \qquad \text{for all } S,T\in 2^E.
\]
We write~$\phi(S)\prec \phi(T)$ instead of~$S\prec_\phi T$ for simplicity. 
We further stick to this order when computing the argument of the maximum or minimum of a function $\phi\colon E\to \RR$, so that
\[
    \arg\max\{\phi(T) \mid T\in E\} = S \quad \Longleftrightarrow \quad \phi(T)\prec \phi(S) \text{ for all } T\in 2^E\setminus S
\]
and analogously for $\arg\min\{\phi(T) \mid T\in E\}$.

An \emph{independence system} is a pair~$(E,\indfam)$ such that~$E$ is a finite set, called the \emph{ground set}, and~$\indfam$ is a non-empty and downwards closed collection of subsets of~$E$; i.e., is such that~$\emptyset \in \indfam$ and, for each~$T\subseteq S \subseteq E$, if~$S\in \indfam$ then~$T\in \indfam$.
We often write~$E=[m]$ for~$m \coloneqq|E|$; and associate $E$ with a set of \emph{agents}. We further refer to the subsets of~$E$ that belong to~$\indfam$ as \emph{independent sets}, which correspond to the eligible sets in our context, and to the remaining subsets of~$E$ as \emph{dependent sets}.
For an independence system~$(E,\indfam)$, we let~$r \colon 2^E \to \NN_0$, defined as~$r(S) \coloneqq \max\{ |T|\mid T\in 2^S \cap \indfam\}$ for all~$S\in 2^E$, denote the \emph{rank} function of the system, i.e., the function that maps each subset of~$E$ to the cardinality of its largest independent subset.
For an independence system~$(E,\indfam)$, we further denote by \mbox{$g\coloneqq\min\{|S|\mid S\in 2^E \setminus \indfam\}$} the size of its smallest dependent set, commonly named the \emph{girth} of the system. 

For an independence system~$(E,\indfam)$, a \textit{score matrix} is a matrix~$W\in \RR^{E\times E}_{+}$ such that~$w_{ii}=0$ for every~$i\in E$. For each~$i,j\in E$ with~$i\neq j$, the entry~$w_{ij}$ represents the score assigned by agent~$i$ to agent~$j$; note that this matrix is binary when agents are only allowed to approve or disapprove others.
We say that~$W$ is a \mbox{\emph{$d$-sparse} matrix} for some~$d\in \NN$ if $|\{j\in E \mid w_{ij}>0\}| \leq d$ for every~$i\in E$.
For a score matrix~$W$ and an agent~$i\in E$, we write~$W_{-i}$ for the matrix in $\RR^{E \times E}_{+}$ obtained by replacing the~$i$-th row from~$W$ by zeros.
For an independence system~$(E,\indfam)$, a score matrix~$W$, and a subset~$S\subseteq E$, we write $\indfam|_S\coloneqq \{I\cap S \subseteq E \mid I \in \cI \}$ for the independence family restricted to~$S$.
For a score matrix~$W$ and sets of agents~$S,T \in 2^E$, we let $\score_{W,S}(T) \coloneqq \sum_{j\in T}\sum_{i\in S} w_{ij}$ be the \emph{total score} of set~$T$ from agents in~$S$; when~$T=\{j\}$ for some~$j\in E$ we write~$\score_{W,S}(j)$ instead of~$\score_{W,S}(\{j\})$ and when~$S=E$ we omit the subindex~$S$.
We also omit the dependence on~$W$ when it is clear from context. 
An instance is given by a tuple~$\Gamma = (E,\indfam,W)$ where~$(E,\indfam)$ is an independence system and~$W$ is a score matrix. 
For an instance~$\Gamma = (E,\indfam,W)$, we let $\OPT(\Gamma) = \arg\max \{ \score_W(T) \mid T\in \indfam\}$ denote an independent set with maximum score. 

A \emph{selection mechanism} is given by a family of functions~$\MEC$, one for each independence system~$(E,\indfam)$, that map an instance $\Gamma = (E,\indfam,W)$ to a probability distribution over~$2^E$ such that $\sum_{S\in \indfam} \MEC_S(\Gamma)=1$, where~$\MEC_S(\Gamma)$ denotes the probability assigned by~$\MEC$ to the set~$S$. 
For~$i\in E$ we also let $\MEC_i(\Gamma) \coloneqq \sum_{S\in \indfam : i\in S} \MEC_S(\Gamma)$ denote the probability assigned by the mechanism to agent~$i$.
In a slight abuse of notation, we will use~$\MEC$ to refer to both the mechanism and to individual functions from the family.
A mechanism~$\MEC$ is \textit{deterministic} on a set of instances~$\indins$ if for every $\Gamma = (E,\indfam,W)\in \indins$ and every~$S\subseteq E$ we have $\MEC_S(\Gamma) \in \{0,1\}$.
A mechanism~$\MEC$ is \emph{impartial} on a set of instances~$\indins$ if, on this set of instances, the scores assigned by an agent do not influence its selection, \ie if for every pair of instances $\Gamma = (E,\indfam,W)$ and $\Gamma' = (E,\indfam,W')$ in~$\indins$ and every agent \mbox{$i\in E$}, \mbox{$\MEC_i(\Gamma)=\MEC_i(\Gamma')$} holds whenever~$W_{-i}=W'_{-i}$.%
\footnote{We remark that this definition of impartiality takes the underlying independence system as given: Agents can strategize by changing their scores but cannot affect the combinatorial constraints.}
A mechanism~$\MEC$ is \mbox{\emph{$\alpha$-optimal}} on a set of instances~$\indins$, for~$\alpha \in [0,1]$, if for every instance $\Gamma = (E,\indfam,W) \in \indins$ the expected total score of the choice of~$\MEC$ differs from the maximum score of an independent set by a factor of at most~$\alpha$, \ie if
\[
    \inf \biggl\{ \frac{\EE[\score_W(\MEC(\Gamma))]}{\score_W(\OPT(\Gamma))}
    \;\Big\vert\; \Gamma = (E,\indfam,W) \in \indins: \score_W(\OPT(\Gamma))>0 \biggr\}\geq \alpha.
\]

In the~$k$-selection setting, \citet{bjelde2017impartial} showed that distributing selection probabilities to sets of size up to~$k$, where the probabilities sum up to~$1$, is equivalent to distributing selection probabilities to agents, where the probabilities sum up to at most~$k$. 
We state an analogous result for matroids in \Cref{sec:matroids}. 
However, one direction of a more general version of this statement holds true for any independence system.

\begin{lemma}
\label{lem:birkhoff-gral}
	Let~$(E,\indfam)$ be an independence system with rank function~$r$. If~$\bfx\in [0,1]^{2^E}$ is a probability distribution such that~$\sum_{S\subseteq E} x_S = \sum_{S\in \indfam} x_S = 1$, then~$p\in [0,1]^E$ defined as~\mbox{$p_i\coloneqq\sum_{S\subseteq E: i\in S} x_S$} for each~$i\in E$ satisfies~$\sum_{i\in S} p_i \leq r(S)$ for every~$S\subseteq E$.
\end{lemma}

\begin{proof}
Let~$(E,\indfam)$,~$r$,~$\bfx$, and~$p$ be as in the statement, and let~$S\subseteq E$ be a subset of~$E$. 
For every subset \mbox{$S'\subseteq E$} with \mbox{$|S\cap S'| > r(S)$}, we have that~$S\cap S' \notin \indfam$, thus~$S'\notin \indfam$ and~$x_{S'}=0$. From this observation, the definition of~$p$, and the fact that~$\sum_{T\in \indfam} x_T = 1$, we conclude that
\begin{equation*}
    \sum_{i\in S} p_i = \sum_{i\in S} \sum_{T\in \indfam : i\in T} x_{T} = \sum_{T\in \indfam} x_{T} \cdot |S\cap T| \leq r(S) \sum_{T\in \indfam} x_{T} = r(S).\qedhere
\end{equation*}
\end{proof}

\section{General Independence Systems}
\label{sec:indep-systems}

We start by considering the problem of impartial selection of agents under combinatorial constraints given by arbitrary independence systems.

In their seminal paper, \citet{alon2011sum} introduced a simple impartial mechanism for selecting up to~$k$ agents: 
Assign the agents to one of two sets~$P_1$ or~$P_2$ uniformly at random, and choose the subset of up to~$k$ agents from~$P_2$ with the highest total score from~$P_1$. 
Impartiality is evident: Only agents in~$P_2$ are eligible for selection and their votes do not influence the outcome. 
\citeauthor{alon2011sum} further showed that this mechanism is \mbox{\smash{$\frac{1}{4}$}-optimal}. 
Each agent of the optimal subset belongs to~$P_2$ with probability~$\smash{\frac{1}{2}}$, and each agent voting for it belongs to~$P_1$ with the same probability. 
Consequently, the expected observed score of the agents in the optimal subset is~$\frac{1}{4}$ times their total score, so the mechanism selects, in expectation, a set with at least this score.

The analysis relying solely on the expected observed score of the optimal subset to bound the mechanism's performance makes it highly robust. 
In particular, it maintains both impartiality and \mbox{$\frac{1}{4}$-optimality} when applied to instances with selection constraints given by independence systems. 
\Cref{alg:twopar} presents a description of the mechanism in this general setting; we call it \nalgtwopar~and denote its outcome for an instance~$\Gamma$ by~$\algtwopar(\Gamma)$. 
\begin{algorithm}[h]
	\SetAlgoNoLine
	\KwIn{an instance~$\Gamma=(E,\indfam,W)$.}
	\KwOut{a set~$S\in \indfam$.}
	Assign every agent~$i \in E$ to one of two sets~$P_1$ or~$P_2$, each with probability~$\frac{1}{2}$\;
	{\bfseries return}~$\arg\max \big\{\score_{W,P_1}(T)\mid T\in \indfam \text{ s.t.\ } T\subseteq P_2\big\}$
	\caption{\nalgtwopar~($\algtwopar$)}
	\label{alg:twopar}
\end{algorithm}
\begin{theorem}
\label{thm:gral-indep-lb}
	\nalgtwopar~is an impartial and \mbox{$\frac{1}{4}$-optimal} selection mechanism.
\end{theorem}

\begin{proof}
The fact that \nalgtwopar~is a selection mechanism follows directly from its definition: For any instance~$(E,\indfam,W)$ and any realization of its internal randomness, it returns a set~$S\in \indfam$.

We now show impartiality. We let~$(E,\indfam)$ be an independence system,~$i \in E$ an agent, and~$W,W'$ score matrices such that~$W_{-i} = W'_{-i}$. We write~$\Gamma\coloneqq(E,\indfam,W)$ and \mbox{$\Gamma'\coloneqq(E,\indfam,W')$} for compactness. For a fixed partition~$(P_1,P_2)$ of the set~$E$ and a score matrix~$Z$, we let~$\algtwopar^{(P_1,P_2)}(E,\indfam,Z)$ denote the outcome (set in~$\indfam$) of \nalgtwopar~when~$(P_1,P_2)$ is the realized partition. We observe that, for any partition~$(P_1,P_2)$ of~$E$,
\begin{align*}
	i \in \algtwopar^{(P_1,P_2)}(\Gamma) & \quad \Longleftrightarrow \quad i \in \arg\max \big\{\score_{W,P_1}(T)\mid T\in \indfam \text{ s.t.\ } T\subseteq P_2\big\} \\
	& \quad \Longleftrightarrow \quad i \in \arg\max \big\{\score_{W',P_1}(T)\mid T\in \indfam \text{ s.t.\ } T\subseteq P_2\big\}\\
	& \quad \Longleftrightarrow \quad i \in \algtwopar^{(P_1,P_2)}(\Gamma'),
\end{align*}
where the first and last equivalences are straightforward from the definition of \nalgtwopar~and the second one uses that $\score_{W,P_1}(j)=\score_{W',P_1}(j)$ for every~$j\in E$ whenever~$i\in P_2$ due to~$W_{-i} = W'_{-i}$.
We obtain that, for any fixed partition of the agents into two sets, agent~$i$ is selected for the instance~$\Gamma$ if and only if it is selected for the instance~$\Gamma'$. 
We conclude by observing that the selection probability of agent~$i$ is simply the ratio between the number of~$2$-partitions of~$E$ for which it gets selected and the total number of~$2$-partitions: 
Denoting by~$\mathcal{P}^2_E$ all~$2^{|E|}$ partitions~$(P_1,P_2)$ of the set~$E$, we have that
\begin{align*}
	\algtwopar_i(\Gamma) & = \frac{1}{2^{|E|}} \sum_{(P_1,P_2) \in \mathcal{P}^2_E} \big\llbracket i \in \algtwopar^{(P_1,P_2)}(\Gamma)\big\rrbracket  =  \frac{1}{2^{|E|}} \sum_{(P_1,P_2) \in \mathcal{P}^2_E} \big\llbracket i \in \algtwopar^{(P_1,P_2)}(\Gamma')\big\rrbracket = \algtwopar_i(\Gamma').
\end{align*}

We finally prove that \nalgtwopar\ is \mbox{$\frac{1}{4}$-optimal}. For an instance~$\Gamma=(E,\indfam,W)$, we claim that
\begingroup
\allowdisplaybreaks
\begin{align*}
	\EE_{S \sim \algtwopar(\Gamma)}[\score_W(S)] & = \frac{1}{2^{|E|}} \sum_{(P_1,P_2)\in \mathcal{P}^2_E} \score_W\big(\algtwopar^{(P_1,P_2)}(\Gamma)\big) \\
	& = \frac{1}{2^{|E|}} \sum_{(P_1,P_2)\in \mathcal{P}^2_E} \max \big\{\score_{W,P_1}(T)\mid T\in \indfam \text{ s.t.\ } T\subseteq P_2 \big\} \\
        & \geq \frac{1}{2^{|E|}} \sum_{(P_1,P_2)\in \mathcal{P}^2_E} \score_{W,P_1}(P_2 \cap \OPT(\Gamma)) \\
	& = \sum_{i\in \OPT(\Gamma)} \EE[\score_{W,P_1}(i) ~|~ i\in P_2 ]~\PP[i \in P_2]\\
	& = \sum_{i\in \OPT(\Gamma)} \frac{1}{2} \cdot \frac{1}{2}\score_W(i) \\
	& = \frac{1}{4}\score_W(\OPT(\Gamma)).
\end{align*}
\endgroup
Indeed, the first two equalities follow directly from the definition of \nalgtwopar, the inequality from the fact that~$\OPT(\Gamma)\in \indfam$, and the remaining equalities from direct computations and the fact that every agent~$i \in E$ is assigned to~$P_1$ or~$P_2$ with probability~$\frac{1}{2}$ independently of other agents. 
This concludes the proof.
\end{proof}

We now study upper bounds on the approximation ratio achievable by impartial mechanisms. In terms of general upper bounds, it is known since the work of \citet{alon2011sum} that no impartial mechanism can provide an approximation ratio better than~$\frac{1}{2}$ in all instances. This holds true already in the setting where all scores are~$0$ or~$1$ and the score matrix is \mbox{$1$-sparse}. Despite the broader range of instances considered when allowing more votes with arbitrary scores, this bound remains best known. We state its validity in the following proposition, which will in fact follow as a special case of a later result.

\begin{proposition}[\citet{alon2011sum}]\label{prop:ub-gral}
    Let~$f$ be a selection mechanism that is impartial and \mbox{$\alpha$-optimal} on instances with \mbox{$1$-sparse} binary score matrices. Then~$\alpha \leq \frac{1}{2}$.
\end{proposition}

This result follows immediately from the~$1$-selection setting, equivalent to an independence system in which independent sets are all singletons. If~$|E|=2$ and both agents assign a score of~$1$ to each other, any mechanism must select one of them with a probability of at most~$\smash{\frac{1}{2}}$. If this agent removes its vote, due to impartiality it is still selected with a probability of at most~$\smash{\frac{1}{2}}$. However, this agent is now the only one with a positive total score, so the bound follows.

We now address a more challenging question aimed at better understanding the dependence of the performance of an impartial mechanism on the independence system. We seek to determine, given an independence system, how well an impartial mechanism can perform under the combinatorial constraints imposed by it, in the worst-case among all possible score matrices.

In the~$k$-selection setting, the independence system is entirely characterized by~$k$, and upper bounds parameterized by~$k$ have been extensively studied. In the case of general independence systems, the parameter that plays a similar role to~$k$ in providing upper bounds for a given structure is the girth of the system, i.e., the size of its smallest dependent set. Note that, for~$k$-selection, dependent sets are all sets of size~$k+1$ or larger, thus the girth is equal to~$k+1$.

We provide two upper bounds dependent on the girth of the underlying independence system: one applicable when the score matrix is \mbox{$1$-sparse} and one when it is arbitrary. 
These results even hold when the score matrix is binary, enabling direct comparison with bounds for~$k$-selection.

\begin{theorem}
\label{thm:gral-indep-ub-plu}
	Let~$(E,\indfam)$ be an independence system with girth~$g$. 
    Let~$\MEC$ be a selection mechanism that is impartial and \mbox{$\alpha$-optimal} on the set of instances~$(E,\indfam,W)$ such that~$W$ is \mbox{$1$-sparse} and binary.
    Then~$\alpha \leq 1-\frac{1}{g(g-1)}$.
\end{theorem}
\begin{proof}
    Let~$(E,\indfam)$,~$\MEC$, and~$\alpha$ be as in the statement. 
    Let~$S\subseteq E$ be such that~$S\not\in \indfam$ and~\mbox{$|S| = g$}, whose existence is guaranteed since the system has girth~$g$. 
    We write \mbox{$S=\{1,\ldots, g \}$} for simplicity and we define a score matrix~$W\in \{0,1\}^{E\times E}$ by letting~$w_{i,i+1}=1$ for each~$i\in [g]$ (where $g+1\coloneqq 1$) and~$w_{ij}=0$ for all other pairs of entries~$(i,j)$. 
    We now consider the instance $\Gamma\coloneqq(E,\indfam,W)$. Since~$S\not\in \indfam$, \Cref{lem:birkhoff-gral} implies that
	\[
		\sum_{i\in S} \MEC_j(\Gamma) \leq g-1.
	\]
	By an averaging argument, there exists~$i'\in [g]$ with~$\MEC_{i'}(\Gamma) \leq \frac{g-1}{g}$. 
    We define another score matrix~$W'$, where $w'_{i',i'+1}=0$ and $w'_{ij}=w_{ij}$ for all other pairs of entries~$(i,j)$, 
    and we consider the instance $\Gamma'\coloneqq(V,\indfam,W')$.
    Since $W_{-i'} = W'_{-i'}$, impartiality yields \mbox{$\MEC_{i'}(\Gamma') = \MEC_{i'}(\Gamma) \leq \frac{g-1}{g}$}, and thus
	\[
		\EE_{T\sim \MEC(\Gamma')}[\score_{W'}(T)] = \sum_{i \in S}\MEC_i(\Gamma')\score_{W'}(i) 
        \leq g-2+\frac{g-1}{g} = \bigg(1 - \frac{1}{g(g-1)} \bigg) (g-1).
	\]
	Since~$|S|=g$, the definition of the girth implies that for every~$S'\subsetneq S$ we have that~$S'\in \indfam$, so~$\score_{W'}(\OPT(\Gamma'))=g-1$. We conclude that~$\alpha \leq 1 - \frac{1}{g(g-1)}$.
\end{proof}

We illustrate \Cref{thm:gral-indep-ub-plu} in \Cref{fig:gral-indep-ub-a}. 
It generalizes the best-known upper bound on the approximation ratio of impartial~$k$-selection mechanisms with \mbox{$1$-sparse} score matrices, due to \citet{alon2011sum}. Indeed, this value is $\smash{
	1 - \frac{1}{k(k+1)} = 1-\frac{1}{g(g-1)}}$,
where we have replaced~$g=k+1$ in line with the previous discussion.
    \begin{figure}[tb]
    \begin{subfigure}{0.48\textwidth}
    \centering
    \begin{tikzpicture}[scale=1.2]
    \Vertex[x=2.5, y=3.0, Math, shape=circle, color=black, size=.05, label=\MEC_j(\Gamma) \leq \frac{g-1}{g}, fontscale=1.2, position=right, distance=-.08cm]{A}
    \Vertex[x=2.029, y=3.977, Math, shape=circle, color=black, size=.05]{B} 
    \Vertex[x=0.972, y=4.219, Math, shape=circle, color=black, size=.05]{C}
    \Vertex[x=0.124, y=3.542, Math, shape=circle, color=black, size=.05]{D}
    \Vertex[x=0.124, y=2.458, Math, shape=circle, color=black, size=.05]{E}
    \Vertex[x=0.972, y=1.781, Math, shape=circle, color=black, size=.05]{F}
    \Vertex[x=2.029, y=2.023, Math, shape=circle, color=black, size=.05]{G}
    \Edge[Direct, color=color2, lw=1pt, bend=-15](A)(B)
    \Edge[Direct, color=black, lw=1pt, bend=-15](B)(C)
    \Edge[Direct, color=black, lw=1pt, bend=-15](C)(D)
    \Edge[Direct, color=black, lw=1pt, bend=-15](D)(E)
    \Edge[Direct, color=black, lw=1pt, bend=-15](E)(F)
    \Edge[Direct, color=black, lw=1pt, bend=-15](F)(G)
    \Edge[Direct, color=black, lw=1pt, bend=-15](G)(A)
    \end{tikzpicture}
    \caption{Agent~$j$ is selected with probability at most~$\frac{g-1}{g}$ when it casts the orange vote, which still holds upon deleting it due to impartiality.}
    \label{fig:gral-indep-ub-a}
    \end{subfigure}
    \hfill 
    \begin{subfigure}{0.48\textwidth}
    \centering
    \begin{tikzpicture}[scale=1.2]
    \Vertex[x=2.5, y=3.0, Math, shape=circle, color=black, size=.05, label=\MEC_j(\Gamma) \geq \frac{g-1}{g}\alpha, fontscale=1.2, position=right, distance=-.08cm]{A}
    \Vertex[x=2.029, y=3.977, Math, shape=circle, color=black, size=.05]{B} 
    \Vertex[x=0.972, y=4.219, Math, shape=circle, color=black, size=.05]{C}
    \Vertex[x=0.124, y=3.542, Math, shape=circle, color=black, size=.05]{D}
    \Vertex[x=0.124, y=2.458, Math, shape=circle, color=black, size=.05]{E}
    \Vertex[x=0.972, y=1.781, Math, shape=circle, color=black, size=.05]{F}
    \Vertex[x=2.029, y=2.023, Math, shape=circle, color=black, size=.05]{G}
    \Vertex[x=-1, y=3, Math, shape=circle, color=black, size=.05, label=i, fontscale=1.2, position=left, distance=-.08cm]{Z}
    \Edge[Direct, color=black, lw=1pt](Z)(A)
    \Edge[Direct, color=black, lw=1pt, bend=-10](Z)(B)
    \Edge[Direct, color=black, lw=1pt, bend=15](Z)(C)
    \Edge[Direct, color=black, lw=1pt](Z)(D)
    \Edge[Direct, color=black, lw=1pt](Z)(E)
    \Edge[Direct, color=black, lw=1pt, bend=-15](Z)(F)
    \Edge[Direct, color=black, lw=1pt, bend=10](Z)(G)
    \Edge[Direct, color=color2, lw=1pt](A)(B)
    \Edge[Direct, color=color2, lw=1pt](A)(C)
    \Edge[Direct, color=color2, lw=1pt, bend=-10](A)(D)
    \Edge[Direct, color=color2, lw=1pt, bend=10](A)(E)
    \Edge[Direct, color=color2, lw=1pt](A)(F)
    \Edge[Direct, color=color2, lw=1pt](A)(G)
    \end{tikzpicture}
    \caption{Agent~$j$ is selected with probability at least~$\frac{g-1}{g}\alpha$ when it does not vote, which still holds upon adding the orange votes due to impartiality.}
    \label{fig:gral-indep-ub-b}
    \end{subfigure}
    \caption{Illustration of the upper bounds stated in \Cref{thm:gral-indep-ub-plu,thm:gral-indep-ub-app} for~$g=7$, where votes are represented by arcs. The depicted agents---other than~$i$ in part \subref{fig:gral-indep-ub-b}---form a dependent set, but removing any of them renders the set independent.}
    \label{fig:gral-indep-ub}
    \end{figure}
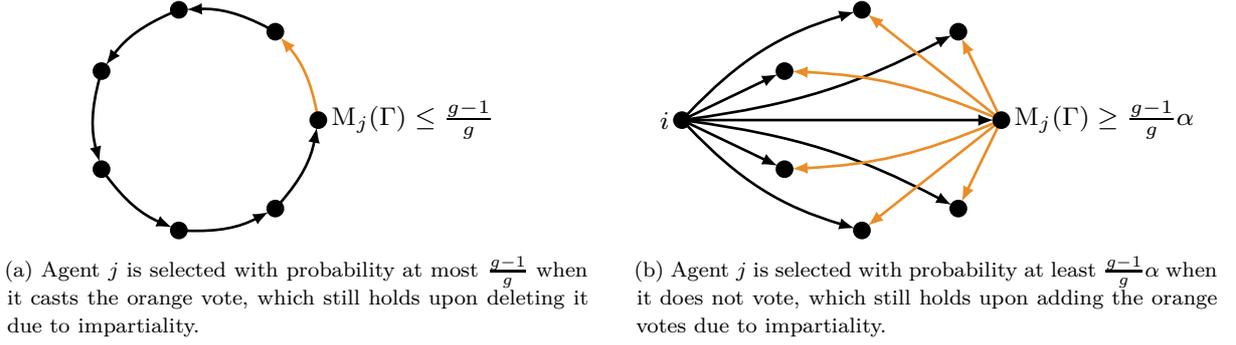

We present a second result concerning upper bounds on the approximation ratio provided by impartial mechanisms on instances with constraints given by independence systems. 
We now also address score matrices that are not \mbox{$1$-sparse}.

\begin{theorem}
\label{thm:gral-indep-ub-app}
	Let~$(E,\indfam)$ be an independence system with girth~$g\leq |E|-1$. Let~$\MEC$ be a selection mechanism that is impartial and \mbox{$\alpha$-optimal} on the set of instances~$(E,\indfam,W)$ such that~$W$ is binary. Then~$\alpha \leq 1-\frac{1}{2g+1}$.
\end{theorem}
\begin{proof}
Let~$(E,\indfam)$,~$\MEC$, and~$\alpha$ be as in the statement. Let~$S\subseteq V$ be such that~$S\not\in \indfam$ and~$|S| = g$, whose existence is guaranteed from the fact that the system has girth~$g$, and let~$i \in E\setminus S$, whose existence is guaranteed from the fact that~$g\leq |E|-1$. 
	We define a score matrix~$W\in \{0,1\}^{E\times E}$ by letting~$w_{ij}=1$ for all~$j\in S$ and~$w_{ij}=0$ for all other pairs of entries~$(j,k)$. 
    We consider an instance~$\Gamma\coloneqq(E,\indfam,W)$.
    The fact that~$\MEC$ is \mbox{$\alpha$-optimal} on~$\Gamma$ implies that
	\[
		\EE_{T\sim \MEC(\Gamma)}[\score_W(T)] = \sum_{j\in S}\MEC_j(\Gamma) \geq \alpha (g-1) = \alpha \score_W(\OPT(\Gamma)),
	\]
	where we used that for every~$S'\subsetneq S$ we have that~$S'\in \indfam$, so~$\score_W(\OPT(\Gamma))=g-1$. By an averaging argument,~$\MEC_{j}(\Gamma) \geq \frac{g-1}{g}\alpha$ for some~$j\in S$; we fix such agent~$j$. We define another score matrix~$W'$, where~$w'_{jk}=0$ for every~$k\in S-j$ and~$w'_{k\ell}=w_{k\ell}$ for all other pairs of entries~$(k,\ell)$,
    and we consider the instance~$\Gamma'\coloneqq(E,\indfam,W')$. We observe that~$\score_{W'}(j) = 1$,~$\score_{W'}(k)=2$ for every~$k\in S-j$, and~$\score_{W'}(\OPT(\Gamma'))=2(g-1)$. As impartiality implies that~$\MEC_{j}(\Gamma) = \MEC_{j}(\Gamma')$, we obtain that 
	\[
		\EE_{T\sim \MEC(\Gamma')}[\score_{W'}(T)] = \sum_{k \in S}\MEC_k(\Gamma')\score_{W'}(k) \leq 2(g-1-\MEC_{j}(\Gamma')) + \MEC_{j}(\Gamma') \leq 2g-2-\frac{g-1}{g}\alpha,
	\]
	where the first inequality follows from the inequality~$\sum_{k\in S} \MEC_k(\Gamma') \leq g-1$ due to \Cref{lem:birkhoff-gral} and the fact that~$S\not\in \indfam$. The equality~$\score_{W'}(\OPT(\Gamma'))=2(g-1)$ and the \mbox{$\alpha$-optimality} of~$\MEC$ yield
	\[
		2(g-1)\alpha \leq 2g-2-\frac{g-1}{g}\alpha \quad \Longrightarrow \quad \alpha \leq 1-\frac{1}{2g+1}.\qedhere
	\]
\end{proof}

\Cref{fig:gral-indep-ub-b} illustrates this upper bound.
In this case, our result constitutes a slightly weaker version of the best-known upper bound on the approximation ratio achievable by \mbox{$k$-selection} mechanisms on~$\calG$, due to \citet{bjelde2017impartial}. This value is given by
\[
1 - \frac{1}{k+2} = 1-\frac{1}{g+1} < 1-\frac{1}{2g+1},
\]
where we have once again substituted~$g=k+1$. 
Despite the difference, attributable to the lack of symmetry in the case of general independence systems, we remark that the asymptotic behavior of both bounds as the girth of the system grows is the same.

\section{Knapsack Constraints}
\label{sec:knapsack}
In this section, we consider independence systems governed by a knapsack constraint, i.e., independence systems of the form~$(E,\indfam)$ with
$\smash{\indfam = \big\{ S\subseteq E \;\big\vert\; \sum_{i\in S} s_i \leq C \big\}}$,
for a given vector of sizes~$\bfs \in \RR^E_{>0}$ and a capacity~$C>0$. 
We assume that~$s_i\leq C$ for all~$i\in E$ throughout this section, as agents not meeting this condition do not belong to any independent set. 
Given~$T\subseteq E$, we use~$s(T) \coloneqq \sum_{i\in T} s_i$ to denote the total size of the agents in~$T$. 
We call instances with independence systems defined this way as \emph{knapsack instances}, and denote them by $(E,\bfs,C,W)$ instead of $(E,\indfam,W)$, since~$\indfam$ is fully defined by~$\bfs$ and~$C$.

The results in this section are based on the classic greedy algorithm for the knapsack problem, so we start with some basic concepts and results in this regard. 
Given an instance~$\Gamma=(E,\bfs,C,W)$ and an agent~$i\in E$, we let~$\rho_{W,s}(i) \coloneqq \frac{\score_W(i)}{s_i}$ denote the \emph{(score) density} of~$i$. 
We sort the agents in~$E$ by density, breaking ties in favor of the total order in $E$. 
Thus, we assume that \mbox{$E=\{i_1(\Gamma), i_2(\Gamma), \ldots, i_m(\Gamma)\}$} with
\[
    \rho_{W,s}(i_1(\Gamma)) \succ \rho_{W,s}(i_2(\Gamma)) \succ \dots \succ \rho_{W,s}(i_m(\Gamma)).
\]
This order is named the \emph{greedy order} of~$E$. 
Let \mbox{$k(\Gamma) \coloneqq \max\big\{\ell\in [m] \mid \sum_{j=1}^{\ell} s_{i_j(\Gamma)} \leq C\big\}$} be the largest index of an agent in the greedy order that fully fits in the knapsack with all preceding agents. 
The \emph{knapsack greedy set}, \mbox{$\algkgreedy(\Gamma) \coloneqq \{i_1(\Gamma), i_2(\Gamma), \ldots, i_{k(\Gamma)}(\Gamma)\}$}, contains the~$k(\Gamma)$ densest agents in the greedy order. 
The \emph{extended knapsack greedy set},
\[
	\algkgreedyex(\Gamma) \coloneqq \begin{cases}
		\algkgreedy(\Gamma) + i_{k(\Gamma)+1}(\Gamma) & \text{if } k(\Gamma) < m,\\
		\algkgreedy(\Gamma) & \text{otherwise,}
		\end{cases}
\]
contains the knapsack greedy set and the next densest agent in the greedy order, if such an agent exists. Observe that, for every instance~$\Gamma=(E,\bfs,C,W)$, we have~$s(\algkgreedy(\Gamma)) \leq C$. For instances with~$s(E)>C$, we also have ~$s(\algkgreedyex(\Gamma)) > C$. When the instance~$\Gamma$ is clear from the context, we omit this argument from the previous definitions.

We now outline certain properties of the sets formed according to the greedy order, which will be used in subsequent proofs. For an instance~$\Gamma=(E,\bfs,C,W)$ and an agent~$i \in E$, we let
\[
	R_i(\Gamma) \coloneqq  \{j \in E\mid \rho_{W,s}(j) \succ \rho_{W,s}(i)\}
\]
denote the set of agents that precede~$i$ in the greedy order.
\begin{lemma}
\label{lem:kgreedy}
	Let~$\Gamma=(E,\bfs,C,W)$ be a knapsack instance. Then,
	\begin{enumerate}[label=(\roman*)]
		\item for every~$i \in E$, we have~$i \in \algkgreedy(\Gamma)$ if and only if~$\sum_{j\in R_i(\Gamma)} s_j + s_i \leq C$;\label[part]{item:condition-kgreedy}
		\item for every~$i \in E$, we have~$i \in \algkgreedyex(\Gamma)$ if and only if~$\sum_{j\in R_i(\Gamma)} s_j \leq C$;\label[part]{item:condition-kgreedyex}
		\item for every~$\ell\in [m]$, denoting~$C'\coloneqq \sum_{j=1}^{\ell} s_{i_j(\Gamma)}$ we have that\label[part]{item:performance-greedy}
		\[
			\score_W\Bigg(\bigcup_{j=1}^{\ell} \{i_j(\Gamma)\}\Bigg) \geq \min\bigg\{ \frac{C'}{C}, 1 \bigg\} \score_W(\OPT(\Gamma)).
		\]
	\end{enumerate}
\end{lemma}
\begin{proof}
    \Cref{item:condition-kgreedy,item:condition-kgreedyex} follow directly from the definition of the knapsack greedy set and the extended knapsack greedy set. 
    To see \cref{item:performance-greedy}, let~$\Gamma=(E,\bfs,C,W)$ be a knapsack instance. A well-known result regarding the optimality of the fractional solution constructed from the greedy order~\citep[e.g.][\S~17.1]{korte18} implies that
    \begin{equation}
        \sum_{j=1}^{k} \score_W(i_j(\Gamma)) + \frac{C - \sum_{j=1}^{k(\Gamma)} s_{i_j(\Gamma)}}{s_{i_{k(\Gamma)+1}(\Gamma)}} \score_W(i_{k(\Gamma)+1}(\Gamma)) \geq \score_W(\OPT(\Gamma)).\label[ineq]{eq:ineq-preformance-greedy-frac}
    \end{equation}
    We now fix~$\ell\in [m]$, and denote~$C'\coloneqq \sum_{j=1}^{\ell} s_{i_j(\Gamma)}$ and~$\Gamma'\coloneqq(E,\bfs,C',W)$. If~$C'\geq C$, 
    \[
        \score_W\Bigg(\bigcup_{j=1}^{\ell} \{i_j(\Gamma)\} \Bigg) = \sum_{j=1}^{k(\Gamma')} \score_W(i_j(\Gamma')) \geq \score_W(\OPT(\Gamma')) \geq \score_W(\OPT(\Gamma)).
    \]
    Indeed, the equality follows directly from the definition of~$\Gamma'$, the first inequality comes from applying \cref{eq:ineq-preformance-greedy-frac} to the modified instance~$\Gamma'$ along with the fact that~$\sum_{j=1}^{k(\Gamma')} s_{i_j(\Gamma')} = C'$, and the last inequality from~$C' \geq C$.
    Conversely, if~$C'< C$, then
    \begingroup
    \allowdisplaybreaks
    \begin{align*}
        \frac{1}{C'}\score_W\Bigg(\bigcup_{j=1}^{\ell} \{i_j(\Gamma)\} \Bigg) & = \frac{1}{C'} \sum_{j=1}^{\ell} \score_W(i_j(\Gamma)) \\
        & \geq \frac{1}{C} \Bigg( \sum_{j=1}^{k(\Gamma)} \score_W(i_j(\Gamma)) + \frac{C - \sum_{j=1}^{k(\Gamma)} s_{i_j(\Gamma)}}{s_{i_{k(\Gamma)+1}(\Gamma)}} \score_W(i_{k(\Gamma)+1}(\Gamma)) \Bigg)\\
        & \geq \frac{1}{C}\score_W(\OPT(\Gamma)).
    \end{align*}
    \endgroup
    Indeed, the first inequality follows because~$\ell \leq k(\Gamma)$ and the agents~$i_1(\Gamma), i_2(\Gamma), \ldots, i_{k(\Gamma)}(\Gamma)$ are sorted by density, and the second one follows from \cref{eq:ineq-preformance-greedy-frac}.
\end{proof}

With the greedy order in mind, we devise mechanisms inspired by the \emph{plurality with runners-up} mechanism proposed by \citet{tamura2014impartial} for single-nomination and its extension to multi-agent nomination by \citet{cembrano2022optimal}. 
A natural randomized variant of plurality with runners-up for \mbox{$k$-selection} selects with probability~$\frac{k}{k+1}$ each of the~$k$ agents with the highest total score and all other agents that would belong to this set upon the deletion of their votes. 
The correctness of this probability assignment is ensured by the fact that at most one such additional agent exists. 
It is also not hard to observe that the mechanism is impartial and provides a \mbox{$\frac{k}{k+1}$-approximation}.

The general idea of this mechanism can, in principle, be applied to the selection of agents under more general constraints: 
Agents from a greedily-constructed subset are chosen with a fixed probability, and this selection extends to additional agents that would be included in this subset if their votes were removed. 
The difficulty, in the case of knapsack constraints, is that the total score of the knapsack greedy set does not consistently approximate the total score of the optimal set.\footnote{For instance, consider~$\Gamma=(E,\bfs,C,W)$ where~$E=\{1,2\}$,~$s_1=1, s_2=M \gg 1$,~$C=M$, and scores~$w_{12}=M-1$ and~$w_{21}=1$. 
Here, the knapsack greedy set~$\algkgreedy(\Gamma) = \{1\}$ has a total score of~$1$, while the optimal set~$\OPT(\Gamma)=\{2\}$ has a total score of~$M-1$.} 
To overcome this limitation, we instead consider the extended knapsack greedy set, which ensures a total score better than that of the optimal set. 
Although this set may not constitute a feasible outcome, we show that it is possible to select every agent in it, as well as those that would belong to this set in a modified instance where their votes have been removed, with a probability of~$\frac{1}{3}$. 

The previous approach can be naturally generalized to $d$-sparse matrices.
We call this mechanism \nalgprukna\ for \mbox{$d$-sparse} instances and denote it by $\algprukna^d$.

\begin{algorithm}[h]
	\SetAlgoNoLine
	\KwIn{a knapsack instance~$\Gamma=(E,\bfs,C,W)$, where~$W$ is \mbox{$d$-sparse}.}
	\KwOut{a vector~$p \in [0,1]^{2^E}$ with~$\sum_{i\in S}s_i\leq C$ for every~$S\subseteq E$ with~$p_S>0$.}
	$S \gets \emptyset$\;
	\For{$i \in E$}{
		\If{$i \in \algkgreedyex(E,\bfs,C,W_{-i})$}{
			$S \gets S +i$\;
		}
	}
	$p \gets 0^{2^E}$\;
	$\ell \gets \arg\min \{\rho_{W,s}(i_j)\mid j\in [m] \text{ s.t.\ } i_j \in S\}$\tcp*{last agent in~$S$ w.r.t.\ to greedy}
        $p_{\{i_\ell\}}\gets \frac{1}{d+2}$\;
        \For(\tcp*[f]{agents receiving non-zero score from ~$i_\ell$}){$j\in E \cap S$ such that $w_{i_\ell j} >0$}{
            $p_{\{j\}} \gets \frac{1}{d+2}$\;
        }
	$T \gets S\setminus (\{i_\ell\} \cup \{j\in E\mid w_{i_\ell j}>0\})$\tcp*{remaining agents in~$S$}
	$p_T \gets \frac{1}{d+2}$\;
	{\bfseries return}~$p$
	\caption{\nalgprukna\ for \mbox{$d$-sparse} instances ($\algprukna^d$)}
	\label{alg:prukna-dsparse}
\end{algorithm}

We obtain the following result.

\begin{theorem}
    \label{thm:knapsack-lb-dsparse}
    For every~$d\in \NN$, \nalgprukna\ for \mbox{$d$-sparse} instances is an impartial and $\frac{1}{d+2}$-optimal selection mechanism on knapsack instances with \mbox{$d$-sparse} score matrices.
\end{theorem}

Note that \Cref{thm:knapsack-lb-dsparse} only represents an improvement over \Cref{thm:gral-indep-lb} for $d=1$.
Thus, we obtain the following consequence for $1$-sparse score matrices.

\begin{corollary}
	\label{thm:knapsack-lb}
	There exists an impartial and \mbox{$\frac{1}{3}$-optimal} selection mechanism on knapsack instances with \mbox{$1$-sparse} score matrices.
\end{corollary}

As outlined above, to achieve this approximation ratio we employ the classic greedy order for the knapsack problem, where the agents are ranked by density. 
We select with probability~$\frac{1}{d+2}$ each agent in the extended knapsack greedy set, as well as each agent that would become part of it by removing its assigned scores from the tallied scores. 
Since we assign probability $\frac{1}{d+2}$ to these agents, to ensure feasibility, we prove that they can be packed into~$d+2$ knapsacks. 
More specifically, letting~$S$ represent this set of agents and~$i_\ell$ be the last agent among them according to the greedy order, we show that~$i_\ell$ can be placed in one knapsack, each potential agent in~$S$ receiving a non-zero score from~$i_\ell$ in another, and all other agents in~$S$ in another knapsack. 
This is depicted in \Cref{fig:knapsack-greedy}.
\begin{figure}[tb]
        \begin{subfigure}{0.48\textwidth}
	\centering
	\begin{tikzpicture}
		\draw[fill = color1trans] (0,0) rectangle node[label={[label distance=0.04cm]90:\small$3$}] {\small$10$} (0.75,0.7);
            \draw[fill = color1trans] (0.75,0) rectangle node[label={[label distance=0.04cm]90:\small$5.5$}] {\small$15$} (2.125,0.7);
		\draw[fill = color2] (2.125,0) rectangle node[label={[label distance=0.04cm]90:\small$10$}] {\small$25$} (4.625,0.7);
		\draw (4.625,0) rectangle node[label={[label distance=0.04cm]90:\small$1$}] {\small$2$} (4.875,0.7);
		\draw[fill = color3trans] (4.875,0) rectangle node[label={[label distance=0.04cm]90:\small$8$}] {\small$10$} (6.875,0.7);
		\draw [very thick, -](2.5,-0.2) -- (2.5,0.9) node[above] {\small$C=10$};
		\draw [very thick, -](5,-0.2) -- (5,0.9) node[above] {\small$2C$};
		\Text[x=-0.35,y=0.87]{\small$s_i$};
		\Text[x=-0.35,y=0.33]{\small$\score(i)$};
		\draw [thick, -Latex](5.5,0) .. controls (5,-.9) and (4,-.9) .. (3.5,0);
		\draw [thick, -Latex](4.75,0) .. controls (4.5,-.4) and (4.25,-.4) .. (4,0);
		\Text[x=4.5,y=-1]{\small$15$};
		\Text[x=4.7,y=-.32]{\small$3$};
        \end{tikzpicture}
        \caption{The agent of size~$1$ cannot decrease the density of the agent of size~$10$ below its own, thus it is not selected.}
        \label{fig:knapsack-greedy-a}
        \end{subfigure}
        \hfill
        \begin{subfigure}{0.48\textwidth}
        \centering
        \begin{tikzpicture}
	       \draw[fill = color1trans] (0,0) rectangle node[label={[label distance=0.04cm]90:\small$3$}] {\small$10$} (0.75,0.7);
                \draw[fill = color1trans] (0.75,0) rectangle node[label={[label distance=0.04cm]90:\small$5.5$}] {\small$15$} (2.125,0.7);
			\draw[fill = color2] (2.125,0) rectangle node[label={[label distance=0.04cm]90:\small$10$}] {\small$25$} (4.625,0.7);
			\draw[fill = color1trans] (4.625,0) rectangle node[label={[label distance=0.04cm]90:\small$1$}] {\small$2$} (4.875,0.7);
			\draw[fill = color3trans] (4.875,0) rectangle node[label={[label distance=0.04cm]90:\small$8$}] {\small$10$} (6.875,0.7);
                \draw [very thick, -](2.5,-0.2) -- (2.5,0.9) node[above] {\small$C=10$};
		\draw [very thick, -](5,-0.2) -- (5,0.9) node[above] {\small$2C$};
            \Text[x=-0.35,y=0.87]{\small$s_i$};
		\Text[x=-0.35,y=0.33]{\small$\score(i)$};
			\draw [thick, -Latex](5.5,0) .. controls (5,-.9) and (4,-.9) .. (3.5,0);
			\draw [thick, -Latex](4.75,0) .. controls (4.5,-.4) and (4.25,-.4) .. (4,0);
			\Text[x=4.5,y=-1]{\small$15$};
			\Text[x=4.74,y=-.32]{\small$10$};
	\end{tikzpicture}
        \caption{The agent of size~$1$ can decrease the density of the agent of size~$10$ below its own, thus it is selected.}
        \label{fig:knapsack-greedy-b}
        \end{subfigure}
	\caption{Two examples of \nalgprukna. 
    The first agents according to the greedy order are depicted as rectangles, with sizes written above and total scores inside. 
    Non-zero scores between depicted agents are represented as arcs along with their scores. 
    Colored agents are selected with probability~$\frac{1}{3}$; colors illustrate the feasibility of the probability assignment as agents of the same color fit into one knapsack. 
    Note that the score assigned by some agent (other than the agent of size $1$) to the agent of size $10$ decreases by $7$ from \subref{fig:knapsack-greedy-a} to \subref{fig:knapsack-greedy-b}, thus keeping the total score of the agent of size 10 at 25.}
	\label{fig:knapsack-greedy}
\end{figure}
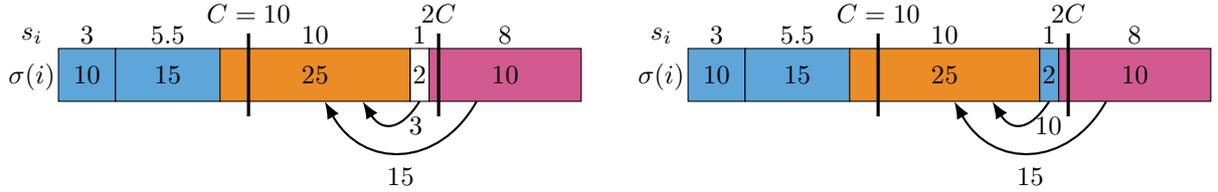

For an instance~$\Gamma$,~$S(\Gamma)$ refers to the set~$S$ defined in \nalgprukna\ for \mbox{$d$-sparse} instances with input~$\Gamma$, and~$\ell(\Gamma)$ denotes the index~$\ell$ defined in this mechanism with input~$\Gamma$. 
Before proving the correctness, impartiality, and \mbox{$\frac{1}{d+2}$-optimality} of the mechanism, we state a simple property regarding the relationship of the extended knapsack greedy set with the set~$S(\Gamma)$.

\begin{lemma}
\label{lem:kgreedyex}
	Let~$d\in \NN$ and~$\Gamma=(E,\bfs,C,W)$ be a knapsack instance such that~$W$ is \mbox{$d$-sparse}. Then, $\algkgreedyex(\Gamma) \subseteq S(\Gamma)$.
\end{lemma}

\begin{proof}
Let~$d$ and~$\Gamma=(E,\bfs,C,W)$ be as in the statement.
Consider any agent~$i\in \algkgreedyex(\Gamma)$. From \cref{item:condition-kgreedyex} of \Cref{lem:kgreedy}, we know that~$\sum_{j\in R_i(\Gamma)} s_j \leq C$. Consider now the modified instance~$\Gamma_{-i} \coloneqq (E,\bfs,C,W_{-i})$. It is immediate that 
\begin{align*}
    \score_{W_{-i}}(i) & = \score_W(i),\\
    \score_{W_{-i}}(j) & \leq \score_W(j) \qquad \text{for every } j \in E-i.
\end{align*}
Therefore,~$\rho_{W_{-i},s}(i) = \rho_{W,s}(i)$ and~$\rho_{W_{-i},s}(j) \leq \rho_{W,s}(j)$ for every~$j \in E-i$. This implies that~$R_i(\Gamma_{-i}) \subseteq R_i(\Gamma)$. Hence,~$\sum_{j\in R_i(\Gamma_{-i})}s_j \leq \sum_{j\in R_i(\Gamma)} s_j \leq C$, so \cref{item:condition-kgreedyex} of \Cref{lem:kgreedy} implies that~$i\in \algkgreedyex(\Gamma_{-i})$. From the definition of~$S(\Gamma)$, we conclude that~$i\in S(\Gamma)$.
\end{proof}

We now have the necessary ingredients to conclude \Cref{thm:knapsack-lb-dsparse}.

\begin{proof}[Proof of \Cref{thm:knapsack-lb-dsparse}]

	Let~$d\in \NN$ be an integer and~$\Gamma=(E,\bfs,C,W)$ be a knapsack instance, where~$W$ is a \mbox{$d$-sparse} score matrix.

	We first prove that the mechanism is well defined, i.e., that it outputs a probability distribution over independent sets. Observe that the set~$S(\Gamma)$ constructed in the mechanism is non-empty due to \Cref{lem:kgreedyex}. Thus,~$\ell(\Gamma)$ is well defined; we call it~$\ell$ for simplicity. Furthermore, we have that
    \[
        \big|\big\{T\subseteq E\mid \algprukna^d_T(\Gamma) > 0\big\}\big| \leq \big|\big\{j\in E \mid w_{i_{\ell}j}>0\big\}\big| + 2 \leq d+2;
    \]
	i.e., at most~$d+2$ sets are assigned strictly positive probability by the mechanism. Since each of them is selected with probability~$\frac{1}{d+2}$, to conclude that the mechanism is well defined it suffices to prove that all these sets are independent, that is, to show the implication
	\[
		\algprukna^d_T > 0 \quad \Longrightarrow \quad s(T)\leq C.
	\]
	Indeed, since~$s_j\leq C$ for every~$j \in E$, this is trivial for the set~$\{i_{\ell}\}$ and for each set~$\{j\}$ with~$j\in E\cap S$ such that $w_{i_{\ell}j}>0$. It remains to prove that~$s(T)\leq C$ for \mbox{$T \coloneqq S(\Gamma) \setminus \big(\{i_{\ell}\} \cup \big\{j\in E \mid w_{i_{\ell} j}>0\big\}\big)$}. 
	
	We let~$\Gamma_{-i_{\ell}} \coloneqq \big(E,\bfs,C,W_{-i_{\ell}}\big)$ denote the modified instance obtained by removing the votes of~$i_{\ell}$. Observe that, due to \cref{item:condition-kgreedyex} of \Cref{lem:kgreedy},~$i_{\ell} \in S(\Gamma)$ implies~$s\big(R_{i_{\ell}}\big(\Gamma_{-i_{\ell}}\big)\big) \leq C$. 
	We claim that~$T \subseteq R_{i_{\ell}}\big(\Gamma_{-i_{\ell}}\big)$. If true, this directly implies~$s(T) \leq C$, as desired. We now show the claim. We fix~$j\in T$, i.e., $j\in S(\Gamma) - i_{\ell}$ with $w_{i_{\ell}j}=0$. Since \mbox{$i_{\ell} = \arg\min\{\rho_{W,s}(i)\mid i\in S(\Gamma)\}$}, we know that 
	\begin{equation}
		\rho_{W,s}(j) \succ \rho_{W,s}(i_{\ell}).\label[ineq]{eq:ineq-densities}
	\end{equation}
	Since~$w_{i_{\ell}j}=0$, we also have $\score_{W_{-i_{\ell}}}(j) = \score_W(j)$. Furthermore, $\score_{W_{-i_{\ell}}}(i_{\ell}) = \score_W(i_{\ell})$ by definition. Replacing in \cref{eq:ineq-densities} then yields $\rho_{W_{-i_{\ell}},s}(j) \succ \rho_{W_{-i_{\ell}},s}(i_{\ell})$, i.e.,~$j\in R_{i_{\ell}}\big(\Gamma_{-i_{\ell}}\big)$.
	
	To prove impartiality, we fix~$i \in E$. We observe that~$\algprukna^d_i(\Gamma) \in \big\{0, \frac{1}{d+2}\big\}$ and that
	\[
		\algprukna^d_i(\Gamma) = \frac{1}{d+2} \quad \Longleftrightarrow \quad i \in S(\Gamma).
	\]
	Moreover, the fact that~$i\in S(\Gamma)$ does not depend on the votes of~$i$ in~$\Gamma$:
	\begin{align*}
	    i\in S(\Gamma) \quad \Longleftrightarrow \quad & i\in \algkgreedyex(E,\bfs,C,W_{-i}) \\
     \Longleftrightarrow \quad & i\in \algkgreedyex(E,\bfs,C,W'_{-i}) \\
     \Longleftrightarrow \quad & i\in S(\Gamma')
	\end{align*}
	for every knapsack instance~$\Gamma'=(V,\bfs,C,W')$ with~$W'_{-i}=W_{-i}$.
	
	Finally, to see \mbox{$\frac{1}{d+2}$-optimality} we claim that
	\begin{align*}
		\sum_{i \in E} \score_W(i) \algprukna^d_i(\Gamma) & = \frac{1}{d+2} \sum_{i\in S(\Gamma)} \score_W(i) \\
		& \geq \frac{1}{d+2} \score_W(\algkgreedyex(\Gamma)) \\
		&  \geq \frac{1}{d+2} \score_W(\OPT(\Gamma)).
	\end{align*}
	Indeed, the first equality comes from the definition of \nalgprukna~and the first inequality from \Cref{lem:kgreedyex}. The last inequality follows from \cref{item:performance-greedy} of \Cref{lem:kgreedy}, since~$s(\algkgreedyex(\Gamma)) > C$ whenever~$\algkgreedyex(\Gamma)\neq E$.
\end{proof}
The tightness of the analysis is evident when~$s(E)\leq C$, as the mechanism returns the set~$E$ with a probability of~$\frac{1}{3}$ (and the empty set with the remaining probability).

Impossibility results for~$1$-selection~\citep{alon2011sum,holzman2013impartial} directly imply that deterministic mechanisms cannot achieve any strictly positive constant approximation ratio in general, even with a sparsity of~$1$, as having agents with sizes equal to~$C$ reduces to this setting. 
However, in the context of knapsack constraints, a natural case in which deterministic mechanisms can provide a non-trivial approximation is one where the capacity significantly exceeds the sizes of the agents. We obtain the following result.

\begin{theorem}
	\label{thm:knapsack-lb-det}
    Let~$d\in \NN$,~$C>0$, and~$0<s_{\max}< \frac{C}{d+1}$. Then, there exists a deterministic selection mechanism that is impartial and \mbox{$\alpha$-optimal} on
    the set of knapsack instances $(E,\bfs,C,W)$ such that~$W$ is a \mbox{$d$-sparse} score matrix and $s_i\leq s_{\max}$ for every~$i\in E$,
    where~$\alpha = 1-\frac{s_{\max}}{C}(d+1)$.
\end{theorem}

The mechanism providing this approximation ratio builds upon those studied by \citet{cembrano2022optimal}. We refer to it as \nalgpruknadet\ for \mbox{$d$-sparse} instances, abbreviate it as~$\algpruknadet^d$, and formally introduce it in \Cref{alg:pruknadet}. Given~$d$ and \mbox{$s_{\max} \geq \max\{s_i \mid i \in E\}$}, the mechanism returns a set containing every agent~$i$ that belongs to the knapsack greedy set with a capacity restricted to~$C-d\cdot s_{\max}$ in the modified instance where its votes have been removed. For an illustration of the mechanism when~$d=1$, see \Cref{fig:pruknadet}.
\begin{algorithm}[h]
	\SetAlgoNoLine
	\KwIn{a knapsack instance~$\Gamma=(E,\bfs,C,W)$, where $W$ is \mbox{$d$-sparse} and~$s_i\leq s_{\max}$ for every~$i \in E$.}
	\KwOut{a subset~$T \subseteq E$ with~$s(T)\leq C$.}
	\KwRet $\{ i \in E \mid i \in \algkgreedy(E,\bfs,C-d\cdot s_{\max},W_{-i})\}$
	\caption{\nalgpruknadet\ for \mbox{$d$-sparse} instances ($\algpruknadet^d$)}
	\label{alg:pruknadet}
\end{algorithm}
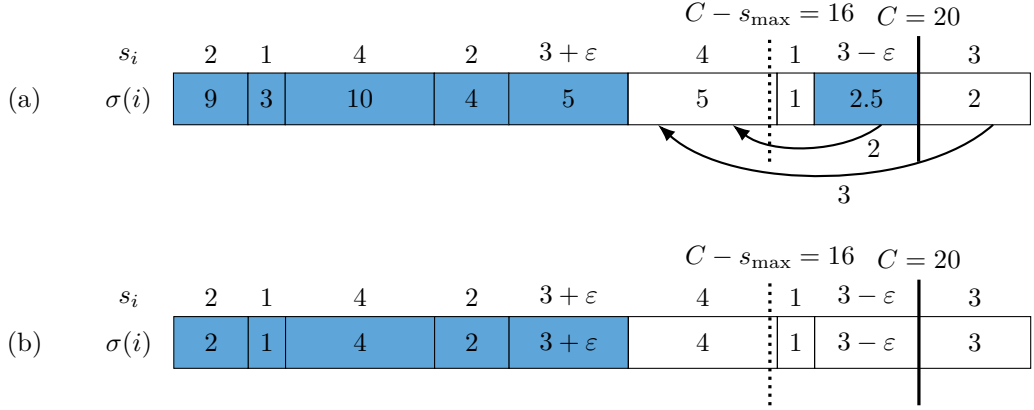
\begin{figure}[tb]
        \begin{subfigure}{\textwidth}
	\centering
	\begin{tikzpicture}[scale=0.98]
		\draw[fill = color1trans] (0,0) rectangle node[label={[label distance=0.1cm]90:\small$2$}] {\small$9$} (1,0.7);
		\draw[fill = color1trans] (1,0) rectangle node[label={[label distance=0.1cm]90:\small$1$}] {\small$3$} (1.5,0.7);
		\draw[fill = color1trans] (1.5,0) rectangle node[label={[label distance=0.1cm]90:\small$4$}] {\small$10$} (3.5,0.7);   
		\draw[fill = color1trans] (3.5,0) rectangle node[label={[label distance=0.1cm]90:\small$2$}] {\small$4$} (4.5,0.7);
		\draw[fill = color1trans] (4.5,0) rectangle node[label={[label distance=0.1cm]90:\small$3+\varepsilon$}] {\small$5$} (6.1,0.7);
		\draw (6.1,0) rectangle node[label={[label distance=0.1cm]90:\small$4$}] {\small$5$} (8.1,0.7);
		\draw (8.1,0) rectangle node[label={[label distance=0.1cm]90:\small$1$}] {\small$1$} (8.6,0.7);
		\draw[fill = color1trans] (8.6,0) rectangle node[label={[label distance=0.1cm]90:\small$3-\varepsilon$}] {\small$2.5$} (10,0.7);
		\draw (10,0) rectangle node[label={[label distance=0.1cm]90:\small$3$}] {\small$2$} (11.5,0.7);
		\draw [very thick, dotted](8,-0.5) -- (8,1.2) node[above] {\small$C-s_{\max}=16$};
		\draw [very thick, -](10,-0.5) -- (10,1.2) node[above] {\small$C=20$};
		\Text[x=-0.6,y=0.93]{\small$s_i$};
		\Text[x=-0.6,y=0.33]{\small$\score(i)$};
            \Text[x=-2,y=0.33]{\small(a)};
		\draw [thick, -Latex](11,0) .. controls (10,-.9) and (7.5,-.9) .. (6.5,0);
		\draw [thick, -Latex](9.5,0) .. controls (9,-.4) and (8,-.4) .. (7.5,0);
		\Text[x=9,y=-.95]{\small$3$};
		\Text[x=9.4,y=-.32]{\small$2$};
		\end{tikzpicture}
            \end{subfigure}
            \par\bigskip 
            \begin{subfigure}{\textwidth}
            \centering
            \begin{tikzpicture}[scale=0.98]
			\draw[fill = color1trans] (0,0) rectangle node[label={[label distance=0.1cm]90:\small$2$}] {\small$2$} (1,0.7);
			\draw[fill = color1trans] (1,0) rectangle node[label={[label distance=0.1cm]90:\small$1$}] {\small$1$} (1.5,0.7);
			\draw[fill = color1trans] (1.5,0) rectangle node[label={[label distance=0.1cm]90:\small$4$}] {\small$4$} (3.5,0.7);   
			\draw[fill = color1trans] (3.5,0) rectangle node[label={[label distance=0.1cm]90:\small$2$}] {\small$2$} (4.5,0.7);
			\draw[fill = color1trans] (4.5,0) rectangle node[label={[label distance=0.1cm]90:\small$3+\varepsilon$}] {\small$3+\varepsilon$} (6.1,0.7);
			\draw (6.1,0) rectangle node[label={[label distance=0.1cm]90:\small$4$}] {\small$4$} (8.1,0.7);
			\draw (8.1,0) rectangle node[label={[label distance=0.1cm]90:\small$1$}] {\small$1$} (8.6,0.7);
			\draw (8.6,0) rectangle node[label={[label distance=0.1cm]90:\small$3-\varepsilon$}] {\small$3-\varepsilon$} (10,0.7);
			\draw (10,0) rectangle node[label={[label distance=0.1cm]90:\small$3$}] {\small$3$} (11.5,0.7);
                \draw [very thick, dotted](8,-0.5) -- (8,1.2) node[above] {\small$C-s_{\max}=16$};
		      \draw [very thick, -](10,-0.5) -- (10,1.2) node[above] {\small$C=20$};
			\Text[x=-0.6,y=0.93]{\small$s_i$};
			\Text[x=-0.6,y=0.33]{\small$\score(i)$};
                \Text[x=-2,y=0.33]{\small(b)};
	\end{tikzpicture}
        \end{subfigure}
	\caption{Two examples of \nalgpruknadet\ for \mbox{$1$-sparse} instances with~$s_{\max}=4$. The notation follows that of \Cref{fig:knapsack-greedy}, votes between depicted agents are represented as arcs along with their scores, and~$\varepsilon>0$ is a small value. Colored agents are selected by the mechanism. Instance~(b), where agents have unit density and are thus sorted by index (not shown in the figure), illustrates the tightness of the analysis, as the guarantee of~$\frac{3}{5}$ implied by \Cref{thm:knapsack-lb-det} is reached by taking~$\varepsilon$ arbitrarily small.}
	\label{fig:pruknadet}
\end{figure}

Impartiality of \nalgpruknadet~follows directly, as the selection condition explicitly ignores the votes of an agent when determining its selection. The approximation ratio follows from the fact that every agent in the knapsack greedy set for a modified capacity of~$C-d\cdot s_{\max}$ remains in this set upon the deletion of its votes, along with \cref{item:performance-greedy} of \Cref{lem:kgreedy}. Finally, the correctness of the mechanism is shown by observing that only agents in the knapsack greedy set for the original instance, which constitute an independent set, are eligible for selection. 

Before proving \Cref{thm:knapsack-lb-det}, we state two simple properties in the following lemma.

\begin{lemma}
\label{lem:pruknadet}
	Let~$d\in \NN$ and~$\Gamma=(E,\bfs,C,W)$ be a knapsack instance such that~$W$ is \mbox{$d$-sparse} and~$s_i\leq s_{\max}$ for every~$i \in E$. Then, for every~$i \in E$ it holds that
	\begin{enumerate}[label=(\roman*)]
		\item if~$i\in \algkgreedy(E,\bfs,C-d\cdot s_{\max},W)$, then~$i\in \algkgreedy(E,\bfs,C-d\cdot s_{\max},W_{-v})$;\label[part]{item:lem-pruknadet-i}
		\item if~$i\in \algkgreedy(E,\bfs,C-d\cdot s_{\max},W_{-v})$, then~$i\in \algkgreedy(E,\bfs,C,W)$.\label[part]{item:lem-pruknadet-ii}
	\end{enumerate}
\end{lemma}

\begin{proof}
	Let~$d$ and~$\Gamma=(E,\bfs,C,W)$ be as in the statement.
	To see \cref{item:lem-pruknadet-i}, consider any agent \mbox{$i\in \algkgreedy(E,\bfs,C-d\cdot s_{\max},W)$}. From \cref{item:condition-kgreedy} of \Cref{lem:kgreedy}, we know that~$s_i + \sum_{j\in R_i(\Gamma)} s_j \leq C-d\cdot s_{\max}$. 
	Consider the modified instance \mbox{$\Gamma_{-i} \coloneqq (E,\bfs,C-d\cdot s_{\max},W_{-i})$}.
	Since
    \begin{align*}
        \score_{W_{-i}}(i) & = \score_W(i),\\
        \score_{W_{-i}}(j) & \leq \score_W(j) \qquad \text{for every } j \in E-i,
    \end{align*}
    it holds that~$\rho_{W_{-i},s}(i)=\rho_{W,s}(i)$ and~$\rho_{W_{-i},s}(j) \leq \rho_{W,s}(j)$ for every~$j \in E-i$. This implies that~$R_i(\Gamma_{-i}) \subseteq R_i(\Gamma)$. Therefore,
    \[
        s_i+\sum_{j\in R_i(\Gamma_{-i})}s_j \leq s_i+\sum_{j\in R_i(\Gamma)} s_j \leq C-d\cdot s_{\max}, 
    \]
    so \cref{item:condition-kgreedy} of \Cref{lem:kgreedy} implies that~$i\in \algkgreedy(\Gamma_{-i})$.
	
	To see \cref{item:lem-pruknadet-ii}, we consider any agent~$i \in E$ with~$i\notin \algkgreedy(\Gamma)$. \Cref{item:condition-kgreedy} of \Cref{lem:kgreedy} implies that
	\begin{equation}
		s_i + \sum_{j\in R_i(\Gamma)} s_j > C.\label[ineq]{eq:cap-violation}
	\end{equation}
	Consider now the instance~$\Gamma_{-i} \coloneqq (E,\bfs,C,W_{-i})$.
	For and agent~$j \in E$ with~\mbox{$w_{ij}=0$}, we know that~$\score_{W_{-i}}(j) = \score_W(j)$ and therefore~$\rho_{W_{-i},s}(j) = \rho_{W,s}(i)$. This implies that \mbox{$R_i(\Gamma)\cap \{j\in E \mid w_{ij}=0\} \subseteq R_i(\Gamma_{-i})$}. Therefore,
	\[
		s_i+\sum_{j\in R_i(\Gamma_{-i})}s_j  \geq s_i + \sum_{j\in R_i(\Gamma)}s_j - \sum_{j\in E: w_{ij}>0}s_j > C - d\cdot s_{\max}.
	\]
	Indeed, the first inequality follows directly from the previous inclusion, while the second one comes from \cref{eq:cap-violation}, the fact that~$W$ is \mbox{$d$-sparse}, and the inequality~$s_i\leq s_{\max}$ for every~$i \in E$. We conclude from \cref{item:condition-kgreedy} of \Cref{lem:kgreedy} that~$i\notin \algkgreedy(E,\bfs,C-d\cdot s_{\max},W_{-v})$.
\end{proof}

We are now ready to prove \Cref{thm:knapsack-lb-det}.

\begin{proof}[Proof of \Cref{thm:knapsack-lb-det}]
	Let~$d\in \NN$ be an integer; we claim the result for \nalgpruknadet\ for \mbox{$d$-sparse} instances. Let~$\Gamma=(E,\bfs,C,W)$ be a knapsack instance such that~$W$ is \mbox{$d$-sparse} and~$s_i\leq s_{\max}$ for every~$i \in E$, with~$s_{\max}(d+1)< C$.
	
	We first show that the mechanism is well defined, i.e., that~$s(\algpruknadet^d(\Gamma)) \leq C$. Indeed, we have that
	\[
		s(\algpruknadet^d(\Gamma)) = \sum_{\substack{i \in E:\\ i\in \algkgreedy(E,\bfs,C-d\cdot s_{\max},W_{-i})}} s_i \leq \sum_{i\in \algkgreedy(E,\bfs,C,W)} s_i \leq C,
	\]
	where the equality follows from the definition of~$\algpruknadet^d(\Gamma)$, the first inequality from \cref{item:lem-pruknadet-ii} of \Cref{lem:pruknadet}, and the last inequality from the definition of the knapsack greedy set.
	
	Impartiality follows directly from the definition of the mechanism, as an agent~$i\in E$ is selected if and only if it belongs to~$\algkgreedy(E,\bfs,C-d\cdot s_{\max},W_{-i})$, which is independent of its votes. Formally, for any knapsack instance~$\Gamma'=(E,\bfs,C,W')$ with a \mbox{$d$-sparse} score matrix~$W'$ and an agent~$i \in E$ such that~$W_{-i}=W'_{-i}$, it holds that
	\begin{align*}
		i \in \algpruknadet^d(\Gamma) & \quad \Longleftrightarrow \quad i\in \algkgreedy(E,\bfs,C-d\cdot s_{\max},W_{-i}) \\
		& \quad \Longleftrightarrow \quad v\in \algkgreedy(E,\bfs,C-d\cdot s_{\max},W'_{-i}) \\
		& \quad \Longleftrightarrow \quad i \in \algpruknadet^d(\Gamma').
	\end{align*}
	
	Finally, to prove the approximation ratio we claim that 
	\begin{equation}
		s(\algkgreedy(E,\bfs,C-d\cdot s_{\max},W)) \geq \min\{ s(E), C-(d+1)\cdot s_{\max}\}.\label[ineq]{eq:kgreedy-reduced}
	\end{equation}
	Indeed, letting~$\Gamma' \coloneqq (E,\bfs,C-d\cdot s_{\max},W)$ denote this instance with reduced capacity, we know from the definition of the knapsack greedy set that~$\algkgreedy(\Gamma') = \bigcup_{j=1}^{k(\Gamma')} \{i_j(\Gamma)\}$.
	If we have~$s(E) \leq C-d\cdot s_{\max}$, then~$\algkgreedy(\Gamma') = E$ and the claim follows directly. Otherwise, we know that 
	\[
		\sum_{j=1}^{k(\Gamma')+1} s_{i_j(\Gamma)} > C-d\cdot s_{\max}.
	\]
	Combining this inequality with the fact that~$s_i\leq s_{\max}$ for every~$i \in E$, we obtain
	\[
			s(\algkgreedy(E,\bfs,C-d\cdot s_{\max},W)) = \sum_{j=1}^{k(\Gamma')} s_{i_j(\Gamma)} = \sum_{j=1}^{k(\Gamma')+1} s_{i_j(\Gamma)} - s_{i_{k(\Gamma')+1}(\Gamma)} > C-(d+1)\cdot s_{\max},
	\] 
	thus \cref{eq:kgreedy-reduced} follows.
	
	We now conclude using \cref{eq:kgreedy-reduced}. 
	From \cref{item:lem-pruknadet-i} of \Cref{lem:pruknadet}, we know that
	\begin{align}
		\score_W(\algpruknadet^d(\Gamma)) & = \sum_{\substack{i \in E:\\ i\in \algkgreedy(E,\bfs,C-d\cdot s_{\max},W_{-i})}} \score_W(i) \nonumber \\
		& \geq \score_W(\algkgreedy(E,\bfs,C-d\cdot s_{\max},W)).\label[ineq]{eq:lb-pruknadet}
	\end{align}
	If~$s(E) \leq C-(d+1)\cdot s_{\max}$, \cref{eq:kgreedy-reduced} implies that
	\[
		\score_W(\algkgreedy(E,\bfs,C-d\cdot s_{\max},W)) = \score_W(E) = \score_W(\OPT(\Gamma)),
	\]
	so the mechanism is trivially optimal due to \cref{eq:lb-pruknadet}. Otherwise, \cref{eq:kgreedy-reduced} and \cref{item:performance-greedy} of \Cref{lem:kgreedy} imply that
	\begin{align*}
		\score_W(\algkgreedy(E,\bfs,C-d\cdot s_{\max},W)) & \geq \frac{C-(d+1)\cdot s_{\max}}{C} \score_W(\OPT(\Gamma)) \\
        & = \bigg( 1 - \frac{s_{\max}}{C}(d+1)\bigg) \score_W(\OPT(\Gamma)),
	\end{align*}
	so we conclude from \cref{eq:lb-pruknadet} once again.
\end{proof}

It is not hard to see that the analysis is again tight. Given~$d$,~$C$, and~$s_{\max}$, consider an instance where all agents have unit density (equal size as indegree). The first agents in the greedy order have a total size of~$C-(d+1)s_{\max}+\varepsilon$ for~$\varepsilon>0$, while the subsequent agents---among which the first one has size~$s_{\max}$---do not vote for previous agents.
The approximation ratio converges to the one stated in \Cref{thm:knapsack-lb-det} as~$\varepsilon$ approaches~$0$; an example is shown in \Cref{fig:pruknadet}b.

\section{Matroid Constraints}
\label{sec:matroids}

We now turn our attention to matroids, one of the most well-studied classes of independence systems where maximal-size independent sets have all the same cardinality.
Instances where the independence system is defined in this way will be called \emph{matroid instances}.
Our main result for general matroids is the following.

\begin{theorem}
\label{thm:matroids-lb}
	There exists an impartial and \mbox{$\frac{1}{2}$-optimal} selection mechanism on matroid instances with \mbox{$1$-sparse} score matrices.
\end{theorem}

As in the knapsack setting, our result is based on the \emph{plurality with runners-up} mechanism studied by \citet{tamura2014impartial} and \citet{cembrano2022optimal}.
The mechanism satisfying \Cref{thm:matroids-lb}, \nalgprugen, is referred to as \algprugen\ and formally introduced in \Cref{alg:prugen}. 
Roughly speaking, \algprugen\ selects with probability~$\frac{1}{2}$ agents that would belong to the maximum score independent set, after removal of the votes they cast.
Additionally, it is well-known that the maximum score independent set can be computed by a greedy algorithm \citep{oxley1992matroid}, thus contributing to the simplicity of \algprugen.

\begin{algorithm}[h]
	\SetAlgoNoLine
	\KwIn{a matroid instance~$\Gamma=(E,\indfam,W)$, where $W$ is \mbox{$1$-sparse}.}
	\KwOut{a vector~$p \in [0,1]^E$ with~$\sum_{e\in S}p_e \leq r(S)$ for every~$S\subseteq E$.}
	$\smash{p \gets 0^E}$\;
	\For{$e \in E$}{
		\If{$e \in \OPT(E,\indfam,W_{-e})$}{
			$p_e \gets \frac{1}{2}$\;
		}
	}
	{\bfseries return}~$p$
	\caption{\nalgprugen~($\algprugen$)}
	\label{alg:prugen}
\end{algorithm}%%

Impartiality of \nalgprugen\ is straightforward, as the selection of an agent~$e\in E$ depends on whether $e \in \OPT(E,\indfam,W_{-e})$ or not, which explicitly ignores the votes of this agent. The approximation ratio is not hard to see either, as we show that for an instance~$\Gamma$, all agents in~$\OPT(\Gamma)$ satisfy the former condition and are thus selected with probability~$\frac{1}{2}$. 
The correctness of the mechanism, in that it induces a probability distribution over independent sets, constitutes the most challenging claim.

To establish the correctness of the mechanism, we first need some basic definitions and results for matroids. 
For a matroid~$(E,\indfam)$ with rank function~$r\colon 2^E \to \NN_0$ and a subset~$S\subseteq E$, the \emph{span} of~$S$ is defined as
\[
	\spn(S) \coloneqq \{e \in E\mid r(S+e) = r(S)\},
\]
i.e., the set of all agents that, when added to~$S$, do not increase its rank.
Similarly, a base of $(E,\cI)$ is an independent set of maximum size.
The following lemma keeps track of some matroid-theoretic properties that will be used throughout.
In particular, it states that, for any instance $\Gamma=(E,\indfam,W)$ and any agent~$e\in E$, the union of maximum weight independent sets when an agent is omitted from the input is no larger than twice the original maximum weight independent set. 

\newpage
\begin{lemma}
\label{lem:matroid-basics}
Let~$(E,\indfam,W)$ be a matroid instance. Then,
\begin{enumerate}[label=(\roman*)]
    \item for every pair of bases~$B,B' \subseteq E$ and every~$x \in B\setminus B'$, there exists~$y \in B'\setminus B$ such that~$B'+x-y$ and~$B+y-x$ are bases.\label[part]{item:strong-basis-exchange}
	\item for every~$S,T\subseteq E$ with~$S\subseteq T$, it holds that~$\spn(S) \subseteq \spn(T)$;\label[part]{item:span-subset}
	\item $e_i(\Gamma)\in \OPT(\Gamma)$ if and only if~$e_i(\Gamma) \notin \spn(\{e_j(\Gamma) \mid j\in [i-1]\})$, for every $i \in [m]$;\label[part]{item:greedy-span}
    \item $\OPT(\Gamma)$ is a unique maximum weight basis when breaking ties according to the total order on $2^E$; \label[part]{item:OPT-is-basis}
	\item 
    Let~$B\subseteq E$ be the unique maximum weight basis of~$(E,\cI)$ with weights given by~$\sigma_W$. 
    For~$e \in E$, we consider~$B^{-e}\subseteq E-e$, the unique maximum weight basis of $(E-e,\cI|_{E-e})$ with weights given by~$\sigma_W$. 
    Then 
 $\big| \bigcup_{e \in E} B^{-e} \big| \leq 2 |B|$.\label[part]{item:greedy-runners-up}
\end{enumerate}	
\end{lemma}

\begin{proof}
\Cref{item:strong-basis-exchange,item:span-subset} are well-known properties of matroids, while \cref{item:greedy-span,item:OPT-is-basis} are a direct consequence of the fact that the greedy algorithm computes an optimal solution; see, e.g., \citet{oxley1992matroid}.
Hence, we only prove \cref{item:greedy-runners-up}.

We begin by claiming that~$e \in B$ implies that $|B^{-e}\Delta B|\leq 2$.
We consider two cases. 

\noindent\emph{Case 1:}~$r(E-e)<r(E)$.  In this case,~$e$ is in every basis of~$(E,\cI)$.
In particular,~$B^{-e} + e$ is a basis of~$E$. 
Since~$B^{-e}$ are~$B$ optimal in~$(E-e,\cI)$ and~$(E,\cI)$ respectively, this implies that~$B=B^{-e}+e$.
Hence, we obtain that $|B\Delta B^{-e}|\leq 1$ and the claim follows.

\noindent\emph{Case 2:} $r(E-e)=r(E)$.
Suppose, for the sake of contradiction, that $|B^{-e}\Delta B| > 2$.
We observe the following:
\begin{equation}
|B^{-e}\setminus B|=|B^{-e}|-|B^{-e} \cap B|=  |B|-|B^{-e} \cap B| = |B \setminus B^{-e}| . \label[eqn]{eq:same-size-difference}
\end{equation}    
Combining \cref{eq:same-size-difference} with the fact that $|B^{-e}\Delta B| > 2$, we obtain that $|B^{-e}\setminus B|= |B\setminus B^{-e}|\geq 2$.
Thus, we can choose $e^1 \in B \setminus B^{-e}$ that is not~$e$.
By \cref{item:strong-basis-exchange}, we obtain that there is $e^2 \in B_{-e} \setminus B$ such that $B^1 \coloneqq B - e^1 + e^2$ and $B^2 \coloneqq B^{-e} - e^2 + e^1$ are bases.
Optimality of~$B$ and~$B^{-e}$ lets us obtain that $\sigma_W (B^1) \leq \sigma_W (B)$ and $\sigma_W (B^2) \leq \sigma_W (B_{-e})$.
Furthermore,
\[
\sigma_W (B^1) + \sigma_W (B^2) = \sigma_W(B) + \sigma_W (B^{-e}),
\]
implying that $\sigma_W (B) = \sigma_W(B^1)$ and $\sigma_W (B^{-e}) = \sigma_W(B^2)$.
Naturally, in the total order over~$E$,~$e^1$ is either smaller or larger than~$e^2$.
If~$e^1$ is smaller than~$e^2$, this implies that $B^{2} \prec_E B^{-e}$, contradicting the optimality of~$B^{-e}$.
Similarly, if~$e^2$ is smaller than~$e^1$, we obtain that~$B^1 \prec_E B$, contradicting the optimality of~$B$.
Therefore, we obtain that if~$e \in B$, then $|B^{-e} \Delta B|\leq 2$.

To conclude, we observe that (1) if~$e\notin B$, we have that~$B=B^{-e}$, and (2) by the previous claim if~$e \in B$, then $|B^{-e} \setminus B| \leq 1$.
The combination of these two observations yields
\[\Bigg|\bigcup_{e \in E} B^{-e}\Bigg|=\Bigg|\bigcup_{e \notin B} B^{-e} \cup \bigcup_{e \in B} B^{-e} \Bigg| = \Bigg|B \cup \bigcup_{e \in B} (B^{-e}\setminus B) \Bigg| \leq |B|+\sum_{e\in B}|B^{-e}\setminus B| \leq 2|B|,\]
as desired.
\end{proof}

The previous lemma allows us to conclude that, for any set~$S\subseteq E$, the total probability assigned by \nalgprugen\ to agents in this set is no larger than its rank~$r(S)$. 
We now show a second result used in the proof of \Cref{thm:matroids-lb}, stating that we can translate a probability distribution over agents into a probability distribution over independent sets, so that the previous result suffices for correctness; see \Cref{lem:birkhoff-matroids}. 
To prove this, we exploit the \emph{independence polytope} of a matroid, enabling us to represent our probability assignment as a convex combination of extreme points corresponding to independent sets.

\begin{lemma}
\label{lem:birkhoff-matroids}
    Let~$(E,\indfam)$ be a matroid with rank function~$r$ and let~$\bfp \in [0,1]^E$ be such that, for every~$S \subseteq E$,  it holds that $\sum_{e \in S} p_e \leq r(S)$. Then, there exists a probability distribution~$\bfx \in [0,1]^{2^E}$ such that~$\sum_{S\subseteq E}x_S = \sum_{S\in \indfam}x_S = 1$ and~$\sum_{S\subseteq E: e\in S}x_S = p_e$ for all~$e \in E$.
\end{lemma}

\begin{proof}
	Let~$(E,\indfam)$ and~$\bfp$ be as in the statement, and let
    \[
        \calQ \coloneqq \Bigg\{ \bfq \in [0,1]^E \;\Big\vert\; \sum_{e\in S} q_e \leq r(S) \text{ for every } S\subseteq E\Bigg\}.
    \]
    Further, let~$\calQ' \coloneqq \text{conv}(\{ \bfz^S \mid S \in \indfam\})$ be the\emph{ independence set polytope} of the matroid, where~$\text{conv}(A)$ denotes the convex hull of the set~$A$ and~$\bfz^S \in \{0,1\}^E$ is simply defined as
	\[
		z^S_e \coloneqq \begin{cases}
			1 & \text{if } e\in S,\\
			0 & \text{otherwise,}
		\end{cases} 
		\qquad \text{for } e \in E.
	\]
	It is well known that~$\calQ = \calQ'$~\citep{edmonds1970matroids}. Therefore, every point in~$\calQ$ can be written as a convex combination of points in~$\{ \bfz^S \mid S \in \indfam\}$. 
	
	Observe that~$\bfp$ belongs to~$\calQ$ by definition. Thus, there exists~$\bflambda \in [0,1]^{\indfam}$ such that
	\begin{equation}
		\sum_{S\in \indfam}\lambda_S \bfz^S = \bfp, \qquad \sum_{S\in \indfam} \lambda_S = 1.\label[eqs]{eq:convex-comb-matroids}
	\end{equation}
	Let~$\bfx \in [0,1]^{2^E}$ be defined as
	\[
		x_S \coloneqq \begin{cases}
			\lambda_S & \text{if } S\in \indfam,\\
			0 & \text{otherwise,} 
		\end{cases}
		\qquad \text{for } S\subseteq V.
	\]
    From this definition and \cref{eq:convex-comb-matroids}, we have that \mbox{$\sum_{S\subseteq E} x_S = \sum_{S \in \indfam}\lambda_S = 1$}, \mbox{$\sum_{S \in 2^E \setminus \indfam} x_S = 0$}, and
    \[
	\sum_{S\subseteq E: e\in S} x_S = \sum_{S\in \indfam: e\in S} \lambda_S = \sum_{S\in \indfam} \lambda_S z^S_e = p_e \qquad \text{for every } e \in E.
    \]
    We conclude that~$\bfx$ is a probability distribution over~$2^E$, whose support is a subset of~$\indfam$, that induces a distribution over agents in~$E$ equal to~$\bfp$.
\end{proof}

We are now ready to prove \Cref{thm:matroids-lb}.

\begin{proof}[Proof of \Cref{thm:matroids-lb}]
We claim the result for \nalgprugen. 
We start by showing that it is indeed a selection mechanism. 
Let~$\Gamma=(E,\indfam,W)$ be a matroid instance where~$W$ is a \mbox{$1$-sparse} score matrix.
Let~$S\subseteq E$ be a subset of agents; it is well known that ~$M|_S \coloneqq (S,\indfam|_S)$ is also a matroid.%
\footnote{$(S,\cI|_{S})$ is usually called the restriction of~$(E,\cI)$ on~$S$~\citep{oxley1992matroid}.}
Furthermore, we define 
$B^{S}$ as the unique maximum weight independent set of~$M|_S$ with weights given by~$\sigma_{W}(s)$ for~$s \in S$ (with ties broken according to the total order on~$2^E$).
Similarly, we define~$B^S_{-e}$ as the unique maximum weight independent set of~$M|_S$ with weights given by~$\sigma_{W}(s)$ for~$s \in S$ (again, with ties broken according to the total order on~$2^E$).

We claim that
\begin{equation}
	\big\{e\in S\mid e \in B^S_{-e}\big\} \subseteq \bigcup_{e\in S} B^{S-e}. \label[prop]{eq:claim-greedy}
\end{equation}

Indeed, let~$e \in S$ such that $e \in B^{S}_{-e}$.
Since~$W$ is \mbox{$1$-sparse}, there exists a unique~$e^* \in E$ such that~$w_{ee^*} > 0$.
Note that if $e^* \notin S$ or $e^* \in B^S_{-e}$, then $B^S_{-e} = B^{S}$. 
Hence, $e \in B^{S-f}$ for all~$f \in S-e$.
On the other hand, if $e^* \in S \setminus B^S_{-e}$, we have as a consequence that $B^S_{-e} = B^{S-e^*}$. 
Therefore, we have that in any case
$e \in \bigcup_{f \in S} B^{S-f}$,
and \cref{eq:claim-greedy} follows.

We now claim that
\begin{equation}
\{e\in S\mid e \in \OPT(E,\indfam, W_{-e})\} \subseteq \{e\in S \mid e\in B^{S}_{-e}\}  .\label[prop]{eq:greedy-restriction}
\end{equation}
To see this, let~$e \in E$ be an agent. 
We have that
\begin{align*}
	e \in \OPT(E,\indfam, W_{-e}) \cap S \quad \Longleftrightarrow & \quad e \in S \setminus \spn(\{f \in E\mid \sigma_{W_{-e}}(f) \succ \sigma_{W_{-e}}(e)\})\\
	 \quad \Longrightarrow & \quad e \in S \setminus \spn(\{f\in S \mid \sigma_{W_{-e}}(f) \succ \sigma_{W_{-e}}(e)\})\\
	 \quad \Longleftrightarrow & \quad e \in B_{-e}^S.
\end{align*}
Here, the equivalences follow from \cref{item:greedy-span} of \Cref{lem:matroid-basics} and the implication comes from \cref{item:span-subset} of this lemma. 
This finishes the proof of \cref{eq:greedy-restriction}.

We now use the previous properties to observe that
\begingroup
\allowdisplaybreaks
\begin{align*}
	\sum_{e\in S} \algprugen_e(\Gamma) 
    & = \frac{1}{2} |\{e\in S \mid e \in \OPT(E,\indfam, W_{-e})\}| \\
	& \leq  \frac{1}{2} |\{e\in S \mid e \in B_{-e}^S\}| \\
	& \leq \frac{1}{2} \Bigg| \bigcup_{e\in S} B^{S-e} \Bigg| \\
	& \leq \frac{1}{2} \cdot 2 |B^S|\\
	& \leq r(S).
\end{align*}
\endgroup
Indeed, the first equality and last inequality follow from the definition of \nalgprugen~and the fact that~$B^S$ is an independent set contained in~$S$, respectively.
The first inequality follows from \cref{eq:greedy-restriction}, the second one from \cref{eq:claim-greedy}, and the third one from \cref{item:greedy-runners-up} of \Cref{lem:matroid-basics}. 
We conclude from \Cref{lem:birkhoff-matroids} that \nalgprugen~is a selection mechanism.

Impartiality is straightforward from the definition of the mechanism, as the selection probability of each agent~$e \in E$ for an instance~$(E,\indfam,W)$ only depends on whether $e\in \OPT(E,\indfam,W_{-e})$ or not, which is independent of the scores that~$e$ assigns. 
Formally, we consider matroid instances \mbox{$\Gamma=(E,\indfam,W)$} and \mbox{$\Gamma'=(E,\indfam,W')$} that differ only on the scores that a given agent~$e \in E$ assigns, i.e., \mbox{$W_{-e} = W'_{-e}$}. 
By definition, we have that both~$\algprugen_e(\Gamma)$ and~$\algprugen_e(\Gamma')$ are values in~$\big\{0,\frac{1}{2}\big\}$.  
Therefore,
\begin{align*}
    \algprugen_e(\Gamma) = \frac{1}{2} \quad \Longleftrightarrow \quad & v \in \OPT(E,\indfam,W_{-e}) 
    \\ \Longleftrightarrow \quad & v \in \OPT(E,\indfam,W'_{-e}) \\ \Longleftrightarrow \quad & \algprugen_e(\Gamma') = \frac{1}{2}.
\end{align*}

It only remains to prove the approximation ratio.
Again, we consider a matroid instance $\Gamma=(E,\indfam,W)$ that is \mbox{$1$-sparse}. 
We claim that, for every~$e \in \OPT(E,\indfam,W)$, we have $e \in \OPT(E,\indfam,W_{-e})$. 
We note that we can easily conclude the guarantee from the claim.
Indeed, we have that~$\algprugen_e(\Gamma) = \frac{1}{2}$ for every~$e \in \OPT(E,\indfam,W)$, and consequently
\[
	\sum_{e \in E} \sigma_{W}(e) \algprugen_e(\Gamma) \geq \sum_{e \in \OPT(E,\indfam,W)} \frac{1}{2}\sigma_{W}(e) = \frac{1}{2} \OPT(E,\cI,W).
\]
We finish by proving the claim. 
To this end, let~$e \in \OPT(E,\indfam,W)$. 
Using \cref{item:greedy-span} of \Cref{lem:matroid-basics}, we obtain that
\begin{equation}
    e  \notin \spn(\{f \in E \mid  \sigma_{W}(f) \succ \sigma_{W}(e)) \label[eqn]{eq:span-w-wv}.
\end{equation}
Since~$\sigma_{W}(e)=\sigma_{W_{-e}}(e)$ and, for every~$f \in E$,~$\sigma_{W}(f)\geq \sigma_{W_{-e}}(f)$, we obtain that
\[
	\{f \in E\mid \sigma_{W_{-e}}(f) \succ \sigma_{W_{-e}}(e) \} \subseteq \{f \in E\mid \sigma_W(f) \succ \sigma_W(e)\}. 
\]
Thus, \cref{item:span-subset} of \Cref{lem:matroid-basics} implies that
\[
	\spn(\{f \in E\mid \sigma_{W_{-e}}(f)) \succ  \sigma_{W_{-e}}(e) \} \subseteq \spn(\{f \in E\mid \sigma_W(f) \succ \sigma_W(e)\}).
\]
Consequently, applying \cref{eq:span-w-wv} we obtain that $e \notin \spn(\{f \in E\mid \sigma_{W_{-e}}(f) \succ \sigma_{W_{-e}}(e))$. 
By using \cref{item:greedy-span} of \Cref{lem:matroid-basics} we conclude that $e \in \OPT(E,\cI, W_{-e})$ as claimed.
\end{proof}

It is not hard to see that this analysis is tight. 
Given a matroid~$(E,\indfam)$, we can consider a score matrix~$W$ with exactly~$r(E)$ agents with strictly positive total score. 
Each one is selected with a probability of~$\frac{1}{2}$ by \nalgprugen, resulting in the mechanism being only \mbox{$\frac{1}{2}$-optimal}.

The existence of a mechanism guaranteeing a better-than~$\frac{1}{4}$ approximation for matroids beyond \mbox{$1$-sparse} score matrices is a natural open question.
We conclude this section by providing a positive answer to it for binary score matrices in the widely-studied class of \emph{simple graphic matroids}.
Here, the ground set corresponds to the set of edges of a simple undirected graph and the independent sets correspond to its forests. 
Formally,~$(E,\indfam)$ is a simple graphic matroid if there exists a simple undirected graph~$G=(V,E)$ such that~$\indfam = \{S\subseteq E\mid S \text{ is a forest}\}$. 
Furthermore, we write~$(E,V,W)$ for the instances where~$(E,\indfam)$ is a simple graphic matroid (as~$\indfam$ is fully defined by the graph~$(V,E)$), and refer to them as \emph{simple graphic matroid instances}.

\begin{theorem}
\label{thm:graphic-matroids-lb}
    There exists an impartial and \mbox{$\frac{1}{3}$-optimal} selection mechanism on simple graphic matroid instances with binary score matrices.
\end{theorem}

We introduce some notation needed to define the mechanism satisfying \Cref{thm:graphic-matroids-lb}. 
We denote~$n\coloneqq|V|$, and for~$v\in V$ we write $\delta(v)\coloneqq\{e\in E: v\in e\}$ for the edges incident to~$v$ and $N(v)\coloneqq\{u\in V: \{u,v\}\in E\}$ for the neighbors of~$v$. 
For a set~$S$, a permutation of~$S$ is a bijection $\tau\colon [|S|]\to S$. 
For~$i\in [|S|]$ we say that~$\tau(i)\in S$ is the~$i$-th element of the permutation, and for~$e\in S$ we say that~$\tau^{-1}(e)\in [|S|]$ is the position of~$e$ in the permutation. 
We denote by~$\calP_S$ the set of all permutations of~$S$. 
For a permutation~$\tau\in \calP_S$ and a subset~$S'\subseteq S$, we let~$\tau|_{S'}$ denote the restriction of~$\tau$ to~$S'$, i.e.,
\[
    \tau|_{S'}(i) = e' \quad \Longleftrightarrow \quad |\{ e\in S' \mid \tau^{-1}(e) < \tau^{-1}(e')\}| = i-1 \qquad \text{for every } i\in [|S'|] \text{ and } e'\in S'.
\]
For a simple undirected graph~$G=(V,E)$ and a permutation~$\eta \in \calP_V$, we define the \emph{partition induced by~$\eta$} as $\bfP^{\eta}(G)\coloneqq(P^{\eta}_1(G),\ldots,P^{\eta}_{n}(G))$, where 
\[
    P^{\eta}_i(G) \coloneqq \{\{\eta(i),\eta(j)\} \in E\mid j\in [n] \text{ s.t.\ } j>i\} \qquad \text{for every } i\in [n].
\]
In simple words, given~$\eta$, for each~$i\in [n]$ we add to~$P^\eta_i(G)$ all edges incident to the vertex~$\eta(i)$ that have not been added to a previous set. For example, $P^{\eta}_1(G) = \delta(\eta(1))$ and~$P^\eta_n(G) = \emptyset$. We omit the dependence on~$G$ when it is clear from context.

We can now define our mechanism satisfying the conditions of \Cref{thm:graphic-matroids-lb}, denoted by \nalgvertex~and abbreviated as~$\algvertex$; see \Cref{alg:vertex} for a formal description.
We first sample permutations~$\eta\in \calP_V$ and~$\pi\in \calP_E$ uniformly at random and proceed in two stages.
First, we consider the partition ~$\bfP^{\eta}$. 
Given this partition, any subset of~$E$ containing at most one agent from each set~$P^{\eta}_1,\ldots,P^{\eta}_{n}$ constitutes a forest.
Thus, in the second stage, we select a subset of~$E$ by picking at most one agent from each part.
We do so by running a variant of a mechanism that \citet{bjelde2017impartial} called \emph{$k$-partition with permutation} for the $k$-selection setting.
More specifically, 
we define the \emph{observed score} as the total score of an edge from other sets of~$P^\eta$ and from edges in the same set that appear earlier in the permutation~$\pi$.
For each set, we consider the edges of the set one by one in the order given by~$\pi$, maintaining a current candidate (initially the first edge within the set) along with its observed score. 
If the observed score of an edge, excluding a potential vote from the current candidate, exceeds that of the candidate, we take it as the new candidate. 
After visiting all edges, the one that remains the candidate is added to the selected set. 
\begin{algorithm}[h!]
	\SetAlgoNoLine
	\KwIn{a simple graphic matroid instance~$\Gamma=(E,V,W)$.}
	\KwOut{a forest~$S$ of the graph $(V,E)$.}
	$\eta \gets$ permutation of~$V$ taken uniformly at random\;
        $\pi \gets$ permutation of~$E$ taken uniformly at random\;
        $S \gets \emptyset$\;
	\For(\tcp*[f]{we consider each set~$P^{\eta}_i$ separately}){$i \in [n]$}{
            $m_i\gets |P^{\eta}_i|$\;
            \If(\tcp*[f]{if $P^{\eta}_i=\emptyset$, continue with the next set}){$m_i=0$}{
                {\bfseries continue}
            }
		$e' \gets \pi|_{P^{\eta}_i}(1)$ \tcp*{initial candidate is the first agent w.r.t.\ $\pi$}
            $\score' \gets \score_{W,E\setminus P^{\eta}_i}(e')$\tcp*{observed score of current candidate}
		\For{$j \in \{2,\ldots, m_i\}$}{
			$e \gets \pi|_{P^{\eta}_i}(j)$\;
			$R \gets (E\setminus P^{\eta}_i) \cup (\{\pi|_{P^{\eta}_i}(k)\mid k\in [j-1]\}-e')$\tcp*{edges in other sets or before~$e$}
			\If(\tcp*[f]{check for higher observed score}){$\score_{W,R}(e) \geq \score'$}{
                    $\score' \gets \score_{W,R + e'}(e)$\tcp*{update observed score of current candidate}
				$e' \gets e$\tcp*{replace current candidate}
			}
		}
		$S \gets S+e'$\;
	}
	{\bfseries return}~$S$
	\caption{\nalgvertex~($\algvertex$)}
	\label{alg:vertex}
\end{algorithm}

For any realization of its internal randomness, the mechanism returns a forest of the graph~$(V,E)$ due to selecting one edge from each set of the partition: 
If there was a cycle, then two edges incident to the first vertex of the cycle according to the permutation~$\eta$ would be selected from the same set of~$P^\eta$, a contradiction. 
Impartiality follows from the impartiality of the mechanism by \citet{bjelde2017impartial}: 
The votes of an edge are considered only when the edge is no longer eligible.

The proof of the approximation ratio is more demanding and is proven in three steps. 
We begin by introducing some additional notation. 
Given a simple graphic matroid instance $\Gamma=(E,V,W)$ and permutations~$\eta\in \calP_V$  and~$\pi\in \calP_E$, we denote by~$\algvertex^{\eta,\pi}(\Gamma)$ the outcome of~$\algvertex(\Gamma)$ when the permutations realized at the beginning of \Cref{alg:vertex} are~$\eta$ and~$\pi$.
For a simple graphic matroid instance $\Gamma=(E,V,W)$, a fixed permutation~$\eta\in \calP_V$, and a fixed permutation~$\pi\in \calP_E$, we define the \emph{score matrix induced by~$\eta$ and~$\pi$}, $W^{\eta,\pi}\in \RR^{E\times E}_+$, as
\[
    W^{\eta,\pi}_{ef} \coloneqq \begin{cases}
    0 & \text{if there exists } i\in [n] \text{ s.t.\ } e,f\in P^\eta_i \text{ and } \pi^{-1}(e) > \pi^{-1}(f),\\
    W_{ef} & \text{otherwise,} \end{cases} \quad \text{for every } e,f\in E.
\]
In simple terms, this matrix is a copy of~$W$ but sets to~$0$ the score received by an edge from any edge in the same set of the partition induced by~$\eta$ which appears later in the permutation~$\pi$.
We also refer to the scores defined by this matrix as the \emph{observed scores} when the realizations of~$\eta$ and~$\pi$ are clear from the context. 
In the following, whenever expectations over~$\eta$ or~$\pi$ are taken, it is implicit that they are permutations of the sets of vertices and edges, respectively, taken uniformly at random. 
If any of these sets is different from~$V$ or~$E$, we make it explicit.

Our first lemma bounds the expected total score of the set returned by \nalgvertex\ in terms of the expected observed scores. 
In order to prove it, we split the expected total score of the set of edges selected by the mechanism into the expected total score of the edge selected from each set of the partition. We then apply the definition of the score matrix~$W^{\eta,\pi}$ along with a result from previous literature regarding the permutation mechanism \citep{bousquet2014near}.

\begin{lemma}\label{lem:algvertex-actual-observed}
    For every simple graphic matroid instance $\Gamma=(E,V,W)$ with a binary score matrix~$W$, it holds that
    \[
        \EE_{\eta,\pi}[\score_W(\algvertex(\Gamma))] \geq \EE_{\eta,\pi}\Bigg[ \sum_{i=1}^{n} \max\{ \score_{W^{\eta,\pi}}(e) \mid e\in P^{\eta}_i \} \Bigg].
    \]
\end{lemma}

\begin{proof}
    Let~$\Gamma=(E,V,W)$ be a simple graphic matroid instance,~$\eta\in \calP_V$, and~$\pi\in \calP_E$. 
    Let~$e_i$ denote the edge selected by \Cref{alg:vertex} from each non-empty set~$P^\eta_i$ when the permutations realized at the beginning are~$\eta$ and~$\pi$, for~$i\in [n]$. 
    We claim that
    \begin{equation}
        \score_{W^{\eta,\pi}}(e_i) = \max\{ \score_{W^{\eta,\pi}}(e) \mid e\in P^{\eta}_i \} \qquad \text{for every } i\in[n] \text{ with } P^\eta_i\neq \emptyset.\label[eq]{eq:claim-selected-score-i}
    \end{equation}
    If true, this directly implies the result, as
    \[
        \EE_{\eta,\pi}[\score_W(\algvertex(\Gamma))] 
        \geq \EE_{\eta,\pi}\Bigg[\sum_{i\in [n]: P^{\eta}_i\neq \emptyset} \score_{W^{\eta,\pi}}(e_i)\Bigg] = \EE_{\eta,\pi}\Bigg[ \sum_{i=1}^{n} \max\{ \score_{W^{\eta,\pi}}(e) \mid e\in P^{\eta}_i \} \Bigg],
    \]
    where we used the definition of \nalgvertex~and that $W_{ef}\geq W^{\eta,\pi}_{ef}$ for every~$e,f\in E$.

    In order to prove \cref{eq:claim-selected-score-i}, we fix~$i\in [n]$ such that~$P^{\eta}_i\neq \emptyset$.
    For each edge~$f\in P^\eta_i$, let~$e'(f)$ denote the edge~$e'$ (current candidate) as defined in \Cref{alg:vertex} when~$f$ is considered and~$R(f)$ denote the set~$R$ (observed score excluding that from current candidate) as defined in \Cref{alg:vertex} when~$f$ is considered. 
    From the definition of the score matrix~$W^{\eta,\pi}$, for each~$f\in P^\eta_i$ we have
    \begin{equation}
        \score_{W,R(f)+e'(f)}(f) = \score_{W^{\eta,\pi}}(f).\label[eq]{eq:algvertex-observed-score}
    \end{equation}
    That the mechanism selects from~$P^\eta_i$ an edge~$f$ maximizing the score from~$R(f)+e'(f)$, i.e., that
    \begin{equation}
        \score_{W,R(e_i)+e'(e_i)}(e_i) = \max\{\score_{W,R(f)+e'(f)}(f) \mid f\in P^\eta_i\},\label[eq]{eq:selected-score-permutation}
    \end{equation}
    follows directly from an equivalent result for the so-called \emph{permutation mechanism} in the unweighted setting, corresponding to the subroutine that our mechanism runs in each set of the partition without external votes \citep[][Theorem 3]{bousquet2014near}. \Cref{eq:claim-selected-score-i} then follows from \cref{eq:algvertex-observed-score,eq:selected-score-permutation}. 
\end{proof}

The expression on the right-hand side of the previous statement is clearly equivalent to the expected total score of the set output by a simple procedure when the permutations~$\eta$ and~$\pi$ are taken uniformly at random: From each set of the partition~$\bfP^{\eta}$, select an edge with maximum observed score~$W^{\eta,\pi}$. For convenience, we formally describe in \Cref{alg:matroidpar} the procedure that takes a partition and a score function as given and outputs, for each set of the partition, the edge with the maximum score. We call it \nalgmatroidpar\ and denote its output for edges~$E$, a partition~$\bfP$, and a score function $\phi\colon E\to \RR_+$ by~$\algmatroidpar(E,\bfP,\phi)$. 

\begin{algorithm}[h]
	\SetAlgoNoLine
	\KwIn{a set of edges~$E$, a partition~$\bfP=(P_1,\ldots,P_{n})$ of~$E$, and a function $\phi\colon E \to \RR_{+}$.}
	\KwOut{a set~$S\subseteq E$ with $|S\cap P_i| \leq 1$ for every $i\in [n]$.}
        $S \gets \emptyset$\;
	\For{$i \in [n]$}{
		$e'\gets \arg\max\{\phi(e) \mid e\in P_i\}$\;
		$S \gets S+e'$\;
	}
	{\bfseries return}~$S$
	\caption{\nalgmatroidpar~($\algmatroidpar$)}
	\label{alg:matroidpar}
\end{algorithm}

As discussed above, it is straightforward that 
\begin{equation}
     \EE_{\eta,\pi}[\score_{W^{\eta,\pi}}(\algmatroidpar(E,\bfP^{\eta},\score_{W^{\eta,\pi}}))] = \EE_{\eta,\pi}\Bigg[ \sum_{i=1}^{n} \max\{ \score_{W^{\eta,\pi}}(e) \mid e\in P^{\eta}_i \} \Bigg].\label[eq]{eq:equivalence-max-matroidpar}
\end{equation}
The following lemma bounds the left-hand side of this inequality, stating that considering the observed scores instead of the actual scores yields an expected loss of at most~$\frac{2}{3}$.

\begin{lemma}\label{lem:algmatroidpar-observed-actual}
    For every simple graphic matroid instance~$(E,V,W)$, it holds that
    \[
        \EE_{\eta,\pi}[\score_{W^{\eta,\pi}}(\algmatroidpar(E,\bfP^{\eta},\score_{W^{\eta,\pi}}))] \geq \frac{2}{3} \EE_{\eta}[\score_W(\algmatroidpar(E,\bfP^{\eta},\score_W))].
    \]
\end{lemma}

\begin{proof}
    
    Let~$(E,V,W)$ be a simple graphic matroid instance. We first observe that, by linearity of expectation and the law of total expectation,
    \begingroup
    \allowdisplaybreaks
    \begin{align}
        & \phantom{{}={}} \EE_{\eta,\pi}[\score_{W^{\eta,\pi}}(\algmatroidpar(E,\bfP^{\eta},\score_{W^{\eta,\pi}}))] \nonumber\\
        & = \EE_{\eta,\pi} \Bigg[ \sum_{i=1}^{n} \max\{\score_{W^{\eta,\pi}}(e) \mid e\in P^{\eta}_i\} \Bigg] \nonumber\\
        & = \sum_{v\in V} \EE_{\eta,\pi} \big[ \max\big\{\score_{W^{\eta,\pi}}(e) \mid e\in P^{\eta}_{\eta^{-1}(v)}\big\} \big] \nonumber\\
        & \geq \sum_{v\in V} \sum_{e^*\in \delta(v)} \EE_{\eta,\pi} \big[ \score_{W^{\eta,\pi}}(e^*) \mid e^* = \arg\max\big\{\score_W(e) \mid e\in P^{\eta}_{\eta^{-1}(v)}\big\} \big]  \nonumber\\[-10pt]
        & \phantom{= \sum_{v\in V} \sum_{e^*\in \delta(v)} \ }\cdot \PP\big[e^* = \arg\max\big\{\score_W(e) \mid e\in P^{\eta}_{\eta^{-1}(v)}\big\}\big],\label[ineq]{ineq:total-score-matroidpart}
    \end{align}
    \endgroup
    where the term associated to~$v\in V$ in the last expression corresponds to the expected observed score of an agent with a maximum total score from the set associated to~$v$; i.e., from the~$\eta^{-1}(v)$-th set of the partition. 
    In what follows, we focus on showing that, for all~$v\in V$ and~$e^*\in \delta(v)$,
    \begin{equation}
        \EE_{\eta,\pi} \big[ \score_{W^{\eta,\pi}}(e^*) \mid e^* = \arg\max\big\{\score_W(e) \mid e\in P^{\eta}_{\eta^{-1}(v)}\big\} \big] \geq \frac{2}{3} \score_W(e^*).\label[ineq]{claim:exp-score-v}
    \end{equation}
    If true, replacing in \cref{ineq:total-score-matroidpart} yields
    \begingroup
    \allowdisplaybreaks
    \begin{align*}
        \EE_{\eta,\pi}[\score_{W^{\eta,\pi}}(\algmatroidpar(E,\bfP^{\eta},\score_{W^{\eta,\pi}}))]
        & \geq \frac{2}{3} \sum_{v\in V} \sum_{e^*\in \delta(v)}\score_W(e^*) \cdot \PP\big[e^* = \arg\max\big\{\score_W(e) \mid e\in P^{\eta}_{\eta^{-1}(v)}\big\}\big]\\
        & = \frac{2}{3} \sum_{v\in V}\EE_{\eta}\big[\max\big\{\score_W(e) \mid e\in P^{\eta}_{\eta^{-1}(v)}\big\}\big] \\
        & = \frac{2}{3}\EE_{\eta}\Bigg[ \sum_{i=1}^{n} \max\{\score_{W}(e) \mid e\in P^{\eta}_i\} \Bigg]\\
        & = \frac{2}{3}\EE_{\eta}[\score_W(\algmatroidpar(E,\bfP^{\eta},\score_W))],
    \end{align*}
    \endgroup
    as claimed.

    We now prove \cref{claim:exp-score-v}, so we fix~$v\in V$ and~$e^*=\{v,u^*\}\in \delta(v)$. 
    For a permutation~$\eta\in \calP_V$, we write~$A^\eta$ for the event stating that~$e^*$ is the maximum in $P^\eta_{\eta^{-1}(v)}$; i.e.,
    \begin{align}
        A^{\eta}  \coloneqq  & \;\big[ e^* = \arg\max\big\{\score_W(e) \mid e\in P^{\eta}_{\eta^{-1}(v)}\big\} \big]\nonumber\\
         =  & \; [ \eta^{-1}(u) < \eta^{-1}(v) < \eta^{-1}(u^*) \text{ for all } u\in N(v) \text{ s.t.\ } \score_W(\{u,v\}) \succ \score_W(e^*)].\label[eq]{eq:alt-def-A}
    \end{align}
    We observe that
    %\begingroup
    %\allowdisplaybreaks
    \begin{align}
        & \phantom{{}={}} \EE_{\eta,\pi} \big[ \score_{W^{\eta,\pi}}(e^*) \mid A^{\eta} \big] \nonumber\\
        & = \EE_{\eta,\pi} \Bigg[ \sum_{e\in E} W_{ee^*}\llbracket W^{\eta,\pi}_{ee^*}=W_{ee^*} \rrbracket \;\Big\vert\; A^\eta \Bigg] \nonumber\\
        & = \EE_{\eta,\pi} \Bigg[ \sum_{e\in E\setminus \delta(v):W_{ee^*}>0} W_{ee^*}\llbracket W^{\eta,\pi}_{ee^*}=W_{ee^*} \rrbracket + \sum_{e\in \delta(v):W_{ee^*}>0} W_{ee^*}\llbracket W^{\eta,\pi}_{ee^*}=W_{ee^*} \rrbracket \;\Big\vert\; A^\eta \Bigg] \nonumber\\
        & = \sum_{e\in E\setminus \delta(v):W_{ee^*}>0} W_{ee^*}\PP[W^{\eta,\pi}_{ee^*}=W_{ee^*} \mid A^{\eta}] + \sum_{e\in \delta(v):W_{ee^*}>0} W_{ee^*}\PP[W^{\eta,\pi}_{ee^*}=W_{ee^*} \mid A^{\eta}].\label[eq]{eq:observed-max-score-cond}
    \end{align}
    %\endgroup
    Indeed, the first equality follows from the definition of $\score_{W^{\eta,\pi}}(e^*)$ and the score matrix~$W^{\eta,\pi}$, the second one from splitting the observed score of~$e^*$ between the score from edges in~$\delta(v)$ and the score from other edges, and the last one due to linearity of expectation. We now claim that
    \begin{align}
        \PP[W^{\eta,\pi}_{ee^*}=W_{ee^*} \mid A^{\eta}] & = 1 \qquad \text{for every } e\in E\setminus \delta(v) \text{ s.t.\ } W_{ee^*}>0,\label[eq]{eq:claim-prob-others}\\
        \PP[W^{\eta,\pi}_{ee^*}=W_{ee^*} \mid A^{\eta}] & \geq \frac{2}{3} \qquad \text{for every } e\in \delta(v) \text{ s.t.\ } W_{ee^*}>0.\label[ineq]{ineq:claim-prob-adjacent}
    \end{align}
    If true, replacing in \cref{eq:observed-max-score-cond} yields
    \[
        \EE_{\eta,\pi} \big[ \score_{W^{\eta,\pi}}(e^*) \mid A^{\eta} \big] \geq \sum_{e\in E\setminus \delta(v):W_{ee^*}>0} W_{ee^*} + \frac{2}{3}\sum_{e\in \delta(v):W_{ee^*}>0} W_{ee^*} \geq 
        \frac{2}{3} \score_W(e^*),
    \]
    concluding the proof of \cref{claim:exp-score-v}.

    It remains to prove \cref{eq:claim-prob-others,ineq:claim-prob-adjacent}. For the former, let $e\in E\setminus \delta(v)$ be such that~$W_{ee^*}>0$. Observe that the event~$A^\eta$ implies that $e^*\in P^\eta_{\eta^{-1}(v)}$, while $e\notin \delta(v)$ implies that $e\notin P^\eta_{\eta^{-1}(v)}$. Thus, from the definition of the score matrix~$W^{\eta,\pi}$ we conclude that 
    \[
        \PP[W^{\eta,\pi}_{ee^*}=W_{ee^*} \mid A^{\eta}] = 1,
    \]
    as claimed. In order to prove \cref{ineq:claim-prob-adjacent}, we let $e=\{v,v'\}\in \delta(v)$ be such that~$W_{ee^*}>0$ and distinguish two cases. If $\score_W(e) \succ \score_W(e^*)$, we know from \cref{eq:alt-def-A} that, given~$A^\eta$, we must have $\eta^{-1}(v')<\eta^{-1}(v)$ and thus, from the definition of~$W^{\eta,\pi}$,
    \[
        \PP[W^{\eta,\pi}_{ee^*}=W_{ee^*} \mid A^{\eta}] = 1.
    \]
    If $\score_W(e) \prec \score_W(e^*)$, let 
   ~$U$ denote the vertices adjacent to~$v$ such that~$\{u,v\}$ satisfies the opposite inequality, i.e.,
    \[
        U \coloneqq \{u\in N(v) \mid \score_W(\{u,v\}) \succ \score_W(e^*)\}.
    \]
    Then,
    \begingroup
    \allowdisplaybreaks
    \begin{align*}
        & \phantom{{}={}} \PP[W^{\eta,\pi}_{ee^*}=W_{ee^*} \mid A^{\eta}]\\
        & = 1- \PP\big[e\in P^\eta_{\eta^{-1}(v)} \text{ and } \pi^{-1}(e) > \pi^{-1}(e^*) \mid A^{\eta}\big]\\
        & = 1- \PP[\eta^{-1}(v)<\eta^{-1}(v') \text{ and } \pi^{-1}(e) > \pi^{-1}(e^*) \mid \eta^{-1}(u)<\eta^{-1}(v)<\eta^{-1}(u^*) \text{ for every } u\in U]\\
        & = 1-\frac{1}{2} \PP[\eta^{-1}(v)<\eta^{-1}(v') \mid \eta^{-1}(u)<\eta^{-1}(v)<\eta^{-1}(u^*) \text{ for every } u\in U]\\
        & = 1-\frac{1}{2} \cdot \frac{2}{|U|+3}\\
        & = \frac{|U|+2}{|U|+3}\\
        & \geq \frac{2}{3}.
    \end{align*}
    \endgroup
    Indeed, the first inequality follows from the definition of the score matrix~$W^{\eta,\pi}$, the second one from the definition of~$P^\eta$ and \cref{eq:alt-def-A}, the third one from the fact that the permutations~$\eta$ and~$\pi$ are drawn independently, and the following equalities and inequalities from simple computations and the fact that~$|U|\geq 0$. In particular, the fourth inequality uses that~$v'\neq u^*$ due to the graph~$(V,E)$ being simple.\footnote{If we allowed for cycles of length two, we would get~$\frac{1}{2}$ instead of~$\frac{2}{3}$, thus leading to a overall approximation ratio of~$\frac{1}{4}$ that is already attained by \nalgtwopar\ on general independence systems.} This concludes the proof of \cref{ineq:claim-prob-adjacent} and the proof of the lemma.
\end{proof}

The following lemma finally states that \nalgmatroidpar, when run on a partition~$\bfP^{\eta}$ with~$\eta$ taken uniformly at random, provides a \mbox{$\frac{1}{2}$-approximation} of the maximum total score of a forest with respect to the input score matrix, thus providing the last ingredient to conclude the bound provided by \nalgvertex.

\begin{lemma}\label{lem:algpartition-1-2}
    For every simple graphic matroid instance $\Gamma=(E,V,W)$, it holds that
    \[
        \EE_{\eta}[\score_W(\algmatroidpar(E,\bfP^{\eta},\score_W))] \geq \frac{1}{2} \score_W(\OPT(\Gamma)).
    \]
\end{lemma}

\begin{proof}
    Let~$\Gamma=(E,V,W)$ be a simple graphic matroid instance. Let~$F^*\coloneqq\OPT(\Gamma)$ denote the optimal forest. 
    Throughout this proof, for a subset of edges~$E'\subseteq E$ we let
    \[
        V|_{E'} \coloneqq \{v\in e\mid e\in E'\}
    \]
    denote the endpoints of edges in~$E'$,~$n_{E'}\coloneqq \big|V|_{E'}\big|$ denote its size, and $G|_{E'} \coloneqq (V',E')$ denote the subgraph of~$G$ induced by~$E'$. Furthermore, for a function~$\phi\colon:E\to \RR_+$, we denote by~$\phi|_{E'}$ its restriction to the set~$E'\subseteq E$, where~$\phi|_{E'}(e)\coloneqq\phi(e)$ for every~$e\in E'$.
    It is easy to see that
    \begin{align*}
        \EE_{\eta}[\score_W(\algmatroidpar(E,\bfP^{\eta}(G),\score_W))] & = \EE_{\eta}\Bigg[ \sum_{i=1}^{n} \max\{ \score_W(e) \mid e\in P^\eta_i(G)\} \Bigg]\\
        & \geq \EE_{\eta}\Bigg[ \sum_{i=1}^{n} \max\{ \score_W(e) \mid e\in P^\eta_i(G) \cap F^* \} \Bigg]\\
        & = \EE_{\eta\in \calP_{V|_{F^*}}}[\score_W(\algmatroidpar(F^*,\bfP^{\eta}(G|_{F^*}),\score_W|_{F^*}))].
    \end{align*}
    Thus, it suffices to show that
    \[
        \EE_{\eta\in \calP_{V|_{F^*}}}[\score_W(\algmatroidpar(F^*,\bfP^{\eta}(G|_{F^*}),\score_W|_{F^*}))] \geq \frac{1}{2} \score_W(F^*).
    \]
    We show that, in fact,
    \begin{equation}
        \EE_{\eta\in \calP_{V|_F}}[\score_W(\algmatroidpar(F,\bfP^{\eta}(G|_F),\score_W|_F))] \geq \frac{1}{2} \score_W(F) \qquad \text{for every forest } F\subseteq E.\label[ineq]{ineq:claim-forests}
    \end{equation}

    We proceed by (strong) induction on the size of the endpoints of~$F$, namely~$n_F$. The base case~$n_F=1$ is trivial, as it implies~$F=\emptyset$ and thus
    \[
        \EE_{\eta\in \calP_{V|_F}}[\score_W(\algmatroidpar(F,\bfP^{\eta}(G|_F),\score_W|_F))] = 0 = \frac{1}{2} \score_W(F).
    \]
    Let now~$F\subseteq E$ be a forest, and suppose that \cref{ineq:claim-forests} holds for every forest~$F'$ with $n_{F'}\leq n_{F}-1$. 
    Then,
    \begingroup
    \allowdisplaybreaks
    \begin{align}
        & \phantom{{}={}} \EE_{\eta\in \calP_{V|_F}}[\score_W(\algmatroidpar(F,\bfP^{\eta}(G|_F),\score_W|_F))] \nonumber \\
        & = \sum_{v\in V|_{F}} \EE_{\eta\in \calP_{V|_F}}[\score_W(\algmatroidpar(F,\bfP^{\eta}(G|_F),\score_W|_F) \mid v=\eta(1)] \PP[v=\eta(1)]\nonumber \\
        & = \frac{1}{n_F} \sum_{v\in V|_{F}} \Big( \max\{\score_W(e) \mid e\in \delta(v)\cap F\} \nonumber \\[-10pt]
        & \phantom{= \frac{1}{n_F} \sum_{v\in V|_{F}} \big( }\ + \EE_{\eta\in \calP_{V|_{F\setminus \delta(v)}}}\big[\score_W\big(\algmatroidpar(F\setminus \delta(v),\bfP^{\eta}(G|_{F\setminus \delta(v)}),\score_W|_{F\setminus \delta(v)}) \big)\big]\Big),\label[eq]{eq:induction-step}
    \end{align}
    \endgroup
    where we have used the rule of expectations, the definition of \nalgmatroidpar, and the fact that~$P^\eta_{\eta(1)}(G|_F)=\delta(\eta(1))$ for any permutation~$\eta\in \calP_{V|_F}$.
    We note at this point that, for every~$v\in V|_F$, we have $V|_{F\setminus \delta(v)} \subseteq V|_F -v$ and thus $n_{F\setminus \delta(v)} \leq n_F-1$. Hence, the induction hypothesis implies that
    \begin{equation}
        \EE_{\eta\in \calP_{V|_{F\setminus \delta(v)}}}\big[\score_W\big(\algmatroidpar(F\setminus \delta(v),\bfP^{\eta}(G|_{F\setminus \delta(v)}),\score_W|_{F\setminus \delta(v)}) \big)\big] \geq \frac{1}{2}\score_W(F\setminus \delta(v)).\label[ineq]{ineq:inductive-hyp}
    \end{equation}
    Furthermore, we observe that
    \begin{equation}
        \sum_{v\in V|_F} \max\{\score_W(e) \mid e\in \delta(v)\cap F\} \geq \sum_{e\in F} \score_W(e)\label[ineq]{ineq:charging-arg}
    \end{equation}
    due to a charging argument.
    Indeed, since~$F$ is a forest we can define an injective function $\kappa\colon F\to V|_F$ that maps each edge to one of its endpoints (i.e.,~$\kappa(e)\in e$ for every~$e\in F$), by mapping each edge of each tree to its endpoint that is farthest from the root. 
    Then, 
    \begin{align*}
        \sum_{v\in V|_F} \max\{\score_W(e) \mid e\in \delta(v)\cap F\} & \geq \sum_{v\in \kappa(F)} \max\{\score_W(e) \mid e\in \delta(v)\cap F\} \\
        & \geq \sum_{v\in \kappa(F)} \score_W(\kappa^{-1}(v)) \\
        & = \sum_{e\in F} \score_W(e).
    \end{align*}
    From \cref{ineq:charging-arg} we obtain that
    \begin{equation}
        \sum_{v\in V|_F} \max\{\score_W(e) \mid e\in \delta(v)\cap F\} \geq \sum_{e\in F} \score_W(e) = \frac{1}{2}\sum_{v\in V|_F} \score_W(\delta(v)).\label[ineq]{ineq:max-score-lb}
    \end{equation}
    By combining \cref{eq:induction-step,ineq:inductive-hyp,ineq:max-score-lb}, we conclude that
    \begin{align*}
        \EE_{\eta\in \calP_{V|_F}}[\score_W(\algmatroidpar(F,\bfP^{\eta}(G|_F),\score_W|_F))] & \geq \frac{1}{2n_F}\sum_{v\in V|_F} (\score_W(\delta(v)) + \score_W(F\setminus \delta(v))) = \frac{1}{2}\score_W(F).\qedhere
    \end{align*}
\end{proof}

Combining the previous lemmas immediately yields the \mbox{$\frac{1}{3}$-optimality} of \nalgvertex. 
We now show this formally, along with the proof of the correctness and impartiality of this mechanism.

\begin{proof}[Proof of \Cref{thm:graphic-matroids-lb}]
We claim the result for \nalgvertex.
We first check that it is indeed a selection mechanism on simple graphic matroid instances.
To this end, let~$\Gamma=(E,V,W)$ be a simple graphic matroid instance, where $W$ is a binary score matrix.
We claim that, for all permutations~$\eta\in \calP_V$ and~$\pi\in \calP_E$, the output of the mechanism for these realizations, namely~$\algvertex^{\eta,\pi}(\Gamma)$, is a forest of the graph~$G=(V,E)$. 
Since the mechanism simply draws these permutations uniformly at random and returns $\algvertex^{\eta,\pi}(\Gamma)$, this suffices to conclude. 
To show the claim, fix~$\eta\in \calP_V$ and~$\pi\in \calP_E$. 
Suppose towards a contradiction that $\algvertex^{\eta,\pi}(\Gamma)$ contains a cycle, i.e., that there exists~$k\in [n]$ and vertices $\{v_j \mid j\in [k]\} \subseteq V$ such that~$\{v_j,v_{j+1}\}\in \algvertex^{\eta,\pi}(\Gamma)$ for each~$j\in [k]$ 
(we consider $v_{k+1}\coloneqq v_1$ and $v_0\coloneqq v_k$ for simplicity). 
Let $j' \coloneqq \arg\min\{ \eta^{-1}(v_j) \mid j\in [k]\}$ be the index of the first vertex in this set according to the permutation~$\eta$. 
This implies that both of its incident edges in the cycle are assigned to its set in the partition, i.e., 
$\{\{v_{j'-1},v_{j'}\}, \{v_{j'},v_{j'+1}\}\} \subseteq \algvertex^{\eta,\pi}(\Gamma) \cap P^{\eta}_{\eta^{-1}(v_{j'})}$. 
This contradicts the fact that at most one edge is selected from each set, i.e., $|\algvertex^{\eta,\pi}(\Gamma) \cap P^{\eta}_i| \leq 1$ for all~$i\in [n]$.

Impartiality of the mechanism follows directly from the fact that, for each realization of the permutations~$\eta$ and~$\pi$, the vote of an edge~$e$ only affects 
(1) the selection of edges in other sets and 
(2) the selection of edges in the same set that appear later in the permutation~$\pi$, provided that~$e$ is not the candidate at that point and is, therefore, not eligible for selection anymore. 
For further details, we refer the reader to the proof of the impartiality of the \emph{$k$-partition with permutation} mechanism by \citet[][Theorem 5.1]{bjelde2017impartial}, which, given a partition and a permutation of the agents, proceeds in the same way as \nalgvertex.

Finally, we combine the previous lemmas to conclude the approximation ratio. 
For a simple graphic matroid instance~$\Gamma=(E,V,W)$ with $W$ a binary score matrix, we have that
\begin{align*}
    \EE_{\eta,\pi}[\score_W(\algvertex(\Gamma))] & \geq \EE_{\eta,\pi}[\score_{W^{\eta,\pi}}(\algmatroidpar(E,\bfP^{\eta},\score_{W^{\eta,\pi}}))]\\
    & \geq \frac{2}{3} \EE_{\eta}[\score_W(\algmatroidpar(E,\bfP^{\eta},\score_W))]\\
    & \geq \frac{1}{3} \score_W(\OPT(\Gamma)),
\end{align*}
where the first inequality follows from \Cref{lem:algvertex-actual-observed} and \cref{eq:equivalence-max-matroidpar}, the second one from \Cref{lem:algmatroidpar-observed-actual}, and the last one from \Cref{lem:algpartition-1-2}. 
\end{proof}

The previous analysis is tight: \nalgvertex\ is not better than \mbox{$\frac{1}{3}$-optimal} on instances like the one depicted in \Cref{fig:bad-case-vertex-partition}.
Intuitively, each edge~$\{u,v_j\}$ is selected with non-marginal probability only if it belongs to~$P^\eta_{\eta^{-1}(v_j)}$, which occurs with probability~$\frac{1}{2}$. Given this fact, it is selected with probability~$\smash{\frac{2}{3}}$, as it is certainly selected when~$v_{j+1}$ is not in~$P^\eta_{\eta^{-1}(v_j)}$ (which happens with probability~$\smash{\frac{1}{3}}$ conditioned on the former event), and it is selected with probability~$\smash{\frac{1}{2}}$ otherwise.
We state this in the following proposition.
\begin{figure}
% \begin{wrapfigure}[20]{r}{.43\linewidth}
\centering
	\begin{tikzpicture}
		\draw (0,0) node[circle,inner sep=0, minimum size=0.6cm,draw=black,ultra thick](0){\small$u$};
		\foreach \i/\j in {0/1,1/2,2/3,3/4,4/5,5/6,6/7}{\draw (\i*45:2.5cm) node[circle,inner sep=0, minimum size=0.6cm,draw=black,ultra thick](\j){\small$v_\j$};}
		\foreach \i in {19,20,21,22,23}{
			\draw (\i*360/24:2.5cm) node[circle,inner sep=0, minimum size=0.1cm,fill=black](8){};
			\draw (\i*360/24:1.5cm) node[circle,inner sep=0, minimum size=0.1cm,fill=color1](9){};
		}
		\foreach \i in {1,2,3,4,5,6,7}{\draw[-,draw,ultra thick] (0) edge (\i);}
		\foreach \i/\j in {1/2,2/3,3/4,4/5,5/6,6/7}{\draw[-,draw,ultra thick] (\i) edge (\j);}
		\foreach \i/\j/\k in {1.48/0.55/45, 0.64/1.44/90, -0.55/1.48/135, -1.44/0.64/180, -1.48/-0.55/225, -0.64/-1.44/270}
		{\begin{scope}[shift={(\i,\j)},rotate=\k]
			\draw [thick, -Latex, color1](0.75,0) .. controls (0.25,.4) and (-0.25,.4) .. (-0.75,0);
		\end{scope}}
	\end{tikzpicture}
	\caption{Example of a graph~$(V,E)$ (in black) and (binary) score matrix (represented by light blue arcs) such that \nalgvertex~is only \mbox{$\frac{1}{3}$-optimal} on the instance~$(E,V,W)$ when the number of vertices~$\{v_j\}_{j\in [n-1]}$ tends to infinity.}
\label{fig:bad-case-vertex-partition}
\end{figure}
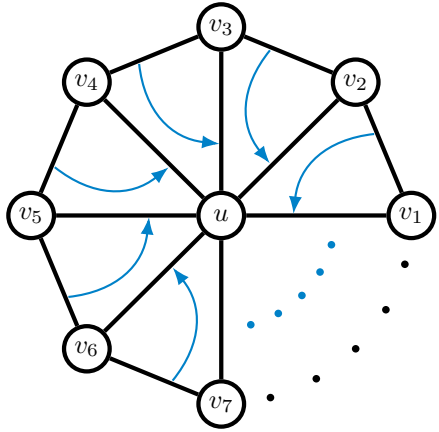
% \end{wrapfigure}

\begin{proposition}\label{prop:tightness-algvertex}
    For every~$\varepsilon>0$, there exists a simple graphic matroid instance~$\Gamma$ with a binary score matrix such that \nalgvertex\ is not \mbox{$\big(\frac{1}{3}+\varepsilon\big)$-optimal} on~$\Gamma$.
\end{proposition}

\begin{proof}
For~$n\in \NN$, we consider a simple graphic matroid instance like the one depicted in \Cref{fig:bad-case-vertex-partition}. 
Specifically, we let~$V\coloneqq \{u\}\cup \{v_j \mid j\in [n-1]\}$,
\[
    E \coloneqq \{\{u,v_j\}\mid j\in [n-1]\} \cup \{\{v_j,v_{j+1}\}\mid j\in [n-1]\},
\]
and a binary score matrix~$W\in \{0,1\}^{E\times E}$ given by
\[
    W_{ef} \coloneqq \begin{cases}
        1 & \text{if there exists } j\in [n-1] \text{ s.t.\ } e=\{v_j,v_{j+1}\} \text{ and } f=\{u,v_j\},\\
        0 & \text{otherwise.}
    \end{cases}
\]
As usual, we write~$v_n \coloneqq v_1$ and~$v_0 \coloneqq v_{n-1}$ for simplicity. 
We also write~$e_j\coloneqq \{v_j,v_{j+1}\}$ and $f_j\coloneqq \{u,v_j\}$ for each~$j\in [n-1]$. 
We consider the instance~$\Gamma=(E,V,W)$. 
As before, we denote by $\algvertex^{\eta,\pi}(\Gamma)$ the outcome of $\algvertex(\Gamma)$ when the permutations realized at the beginning of \Cref{alg:vertex} are~$\eta\in \calP_V$  and~$\pi\in \calP_E$. 
Then, for each~$j\in [n-1]$
\begin{align}
    & \phantom{{}={}} \PP[f_j \in \algvertex^{\eta,\pi}(\Gamma)]\nonumber\\
    & = \PP[f_j \in \algvertex^{\eta,\pi}(\Gamma) \mid \eta^{-1}(v_j)<\eta^{-1}(u)] \PP[\eta^{-1}(v_j)<\eta^{-1}(u)] \nonumber\\
    & \phantom{{}={}} + \PP[f_j \in \algvertex^{\eta,\pi}(\Gamma) \mid \eta^{-1}(v_j)>\eta^{-1}(u)] \PP[\eta^{-1}(v_j)>\eta^{-1}(u)]\nonumber\\
    & = \frac{1}{2} \big( \PP[f_j \in \algvertex^{\eta,\pi}(\Gamma) \mid \eta^{-1}(v_j)<\eta^{-1}(u)] + \PP[f_j \in \algvertex^{\eta,\pi}(\Gamma) \mid \eta^{-1}(v_j)>\eta^{-1}(u)] \big).\label[eq]{eq:algvertex-prob-vj}
\end{align}
We now study the terms of the right-hand side separately. 
To compute the first one, we fix~$j\in [n-1]$ and note that $\eta^{-1}(v_j)<\eta^{-1}(u)$ implies that $f_j \in P^\eta_{\eta^{-1}(v_j)}$. 
We claim that $f_j \in \algvertex^{\eta,\pi}(\Gamma)$ if and only if~$W^{\eta,\pi}_{e_jf_j}=W_{e_jf_j}$. 
The fact that $W^{\eta,\pi}_{e_jf_j}=W_{e_jf_j}$ implies $f_j \in \algvertex^{\eta,\pi}(\Gamma)$ follows from the fact that $\score_W(e)=0$ for every $e\in \delta(v_j)-f_j$. 
For the other implication, note that if $W^{\eta,\pi}_{e_jf_j}=0$, then $\score_{W^{\eta,\pi}}(e)=0$ for every~$e\in P^\eta_{\eta^{-1}(v_j)}$, but~$e_j\in P^\eta_{\eta^{-1}(v_j)}$ is such that~$\pi^{-1}(e_j)>\pi^{-1}(f_j)$ and thus~$f_j$ cannot be selected. 
Therefore,
\begin{align}
    & \phantom{{}={}} \PP[f_j \in \algvertex^{\eta,\pi}(\Gamma) \mid \eta^{-1}(v_j)<\eta^{-1}(u)] \nonumber\\
    & = \PP\big[W^{\eta,\pi}_{e_jf_j}=W_{e_jf_j} \mid \eta^{-1}(v_j)<\eta^{-1}(u)\big]\nonumber\\
    & = \PP\big[W^{\eta,\pi}_{e_jf_j}=W_{e_jf_j} \mid \eta^{-1}(v_{j+1})<\eta^{-1}(v_j)<\eta^{-1}(u)\big] \PP[\eta^{-1}(v_{j+1})<\eta^{-1}(v_j) \mid \eta^{-1}(v_j)<\eta^{-1}(u)]\nonumber\\
    & \phantom{{}={}} + \PP\big[W^{\eta,\pi}_{e_jf_j}=W_{e_jf_j} \mid \eta^{-1}(v_j)<\eta^{-1}(v_{j+1}) \text{ and } \eta^{-1}(v_j)<\eta^{-1}(u)\big]\nonumber\\
    & \phantom{{}={} +} \cdot \PP[\eta^{-1}(v_j)<\eta^{-1}(v_{j+1}) \mid \eta^{-1}(v_j)<\eta^{-1}(u)]\nonumber\\
    & = \frac{1}{3} \PP\big[W^{\eta,\pi}_{e_jf_j}=W_{e_jf_j} \mid e_j\in P^\eta_{\eta^{-1}(v_{j+1})} \text{ and } f_j\in P^\eta_{\eta^{-1}(v_{j})}\big] \nonumber\\
    & \phantom{{}={}} + \frac{2}{3} \PP\big[W^{\eta,\pi}_{e_jf_j}=W_{e_jf_j} \mid \{e_j,f_j\} \subseteq P^\eta_{\eta^{-1}(v_j)}\big]\nonumber\\
    & = \frac{1}{3}\cdot 1 + \frac{2}{3} \cdot \frac{1}{2}\nonumber\\
    & = \frac{2}{3}.\label[eq]{eq:algvertex-prob-vj-before-u}
\end{align}
On the other hand, $\eta^{-1}(v_j)>\eta^{-1}(u)$ for some~$j\in [n-1]$ implies that $f_j \in P^\eta_{u}$. Since at most one edge is selected from~$P^\eta_u$, we have that
\begin{align}
    \sum_{j=1}^{n-1} \PP[f_j \in \algvertex^{\eta,\pi}(\Gamma) \mid \eta^{-1}(v_j)>\eta^{-1}(u)] \leq 1.\label[eq]{eq:algvertex-prob-vj-after-u}
\end{align}
Combining \cref{eq:algvertex-prob-vj,eq:algvertex-prob-vj-before-u,eq:algvertex-prob-vj-after-u}, we obtain that
\[
    \EE[\score_W(\algvertex(\Gamma))] = \sum_{j=1}^{n-1}\PP[f_j \in \algvertex^{\eta,\pi}(\Gamma)] \leq \frac{1}{2} \Bigg(\sum_{j=1}^{n-1} \frac{2}{3} + 1\Bigg) = \frac{1}{3}(n-1)+\frac{1}{2}.
\]
Since $\{f_j\mid j\in [n-1]\}$ is a forest, we have $\score_W(\OPT(\Gamma))=n-1$, and thus
\[
    \frac{\EE[\score_W(\algvertex(\Gamma))]}{\score_W(\OPT(\Gamma))} \leq \frac{1}{3}+\frac{1}{2(n-1)}.
\]
We conclude that, for any~$\varepsilon>0$, we can take~$n$ large enough so that \nalgvertex\ is not \mbox{$\big(\frac{1}{3}+\varepsilon\big)$-optimal} on~$\Gamma$.
\end{proof}

\printbibliography

\end{document}